%% file: main.tex
\documentclass[a4paper,11pt,pointlessnumbers]{scrartcl}

\input{configurations/header}
\input{configurations/commands}

\begin{document}

\title{A unifying method for the design of algorithms canonizing combinatorial objects}

\author{Pascal Schweitzer\footnote{The research leading to these results has received funding from the European Research Council (ERC) under the European Union’s Horizon 2020 research and innovation programme (grant agreement No. 820148).}\\
TU Kaiserslautern\\
\texttt{schweitzer@cs.uni-kl.de}
\and
Daniel Wiebking\\
RWTH Aachen University\\
\texttt{wiebking@informatik.rwth-aachen.de}}

\date{}

\maketitle

\begin{abstract}
We devise a unified framework for the design of canonization algorithms.
Using hereditarily
finite sets, we define a general notion of combinatorial objects
that includes
graphs, hypergraphs, relational structures, codes, permutation
groups, tree decompositions, and so on.

Our approach allows for a systematic transfer of the techniques that have been
developed for isomorphism testing to canonization. We use it to design a
canonization algorithm for combinatorial objects in general. This result gives
new fastest canonization algorithms with an asymptotic running time matching the best known isomorphism algorithm for the following types of objects: hypergraphs, hypergraphs of bounded color class size, permutation groups (up to permutational isomorphism) and codes that are explicitly given (up to code equivalence).
\end{abstract}

\input{articles/intro}

\input{articles/preliminaries}
\input{articles/heredit}

\input{articles/cosetInt}

\input{articles/objectReplacement}

\input{articles/hypergraphCanonizationExp}
\input{articles/multiCoset}
\input{articles/hereditarilyFiniteObjects}

\input{articles/canGenSets}

\input{articles/openQuestions}

\bibliographystyle{alpha}
\bibliography{references}

\addcontentsline{toc}{section}{Bibliography}

\end{document}

%% file: configurations/header.tex
\usepackage[utf8]{inputenc}
\usepackage{amsthm}
\usepackage{amssymb}
\usepackage{enumerate}
\usepackage{mathtools}

\usepackage{fullpage}

\usepackage{lmodern}
\usepackage{MnSymbol}
\usepackage{booktabs}
\usepackage{tabularx}

\usepackage{hyperref}
\usepackage{cleveref}

\usepackage{tikz}
\usetikzlibrary{hobby}
\usepackage{tikz-cd}
\usepackage{bm}

%% file: configurations/commands.tex
\newcommand{\problem}[6][ ]{
\begin{prob}
Compute a function $#2$ with the following properties:\\
\begin{tabular}{@{}p{0.125\linewidth}@{}p{0.875\linewidth}@{}}
Input & #3.\\
Output & A labeling coset $#2#4=\Lambda\leq\lab(V)$ such that:\\
(CL1) & $#2#4=\phi#2#5$ for all $\phi\in\iso(V;V')$.\\
(CL2) & $#2#4=#6\pi$#1for some (and thus for all) $\pi\in\Lambda$.
\end{tabular}
\end{prob}
}

\newcommand{\transitive}[2]{
Let $A^\ord:=A^\rho$ and $\Delta^\ord:=(\Delta\rho)^\rho$.\\
Partition $A^\ord=A^\ord_1\cupdot A^\ord_2$ such that
$|A^\ord_1|=\lfloor|A^\ord|/2\rfloor$ and
$A^\ord_1$ is minimal.\\
\minimal\\
Compute $\Psi^\ord:=\stab_{\Delta^\ord}(A^\ord_1,A^\ord_2)$.\\
\emph{(The group can be computed using a membership test as
stated in the preliminaries.)}\\
Decompose $\Delta^\ord$ into left cosets of $\Psi^\ord$ and write
$\Delta^\ord=\bigcup_{i\in[s]}\delta^\ord_i\Psi^\ord$.\\
Compute $\Delta_i\rho_i:=#1(#2,A,\rho\delta^\ord_i\Psi^\ord)$ for each
$i\in[s]$ recursively.\\
Rename the indices in $[s]$ such that:\\
$(#2,\Delta\rho)^{\rho_1}=\ldots=(#2,\Delta\rho)^{\rho_r} \prec(#2,\Delta\rho)^{\rho_{r+1}}\preceq\ldots\preceq
(#2,\Delta\rho)^{\rho_s}$.\\
Return $\Lambda:=
\langle\Delta_1\rho_1,\ldots,\Delta_r\rho_r\rangle$.\\
\emph{(This is the smallest coset containing these cosets as defined
in the preliminaries.)}
}

\newcommand{\minimal}{
\emph{(Minimal w.r.t. the order ``$\prec$'' defined in
\cref{sec:comb:objs:and:lab:cos}.)}}

\newcommand{\iterated}{
\emph{(This is an iterated instance as defined in Section~\ref{sec:atoms}.)}}

\theoremstyle{plain}
\newtheorem{theo}{Theorem}
\crefname{theo}{Theorem}{Theorems}
\newtheorem{lem}[theo]{Lemma}
\crefname{lem}{Lemma}{Lemmata}
\newtheorem{cor}[theo]{Corollary}
\crefname{cor}{Corollary}{Corollarys}

\theoremstyle{definition}

\newtheorem{exa}[theo]{Example}
\crefname{exa}{Example}{Examples}
\newtheorem{prob}[theo]{Problem}
\crefname{prob}{Problem}{Problems}

\theoremstyle{remark}

\crefname{section}{Section}{Sections}
\crefname{appendix}{Appendix}{Appendices}
\crefname{figure}{Figure}{Figures}

\newenvironment{correctness}{\vspace{0.1cm}\noindent(\textit{Correctness.})}{}
\newenvironment{cl1}{\vspace{0.1cm}\noindent(\textit{CL1.})}{}
\newenvironment{cl2}{\vspace{0.1cm}\noindent(\textit{CL2.})}{}
\newenvironment{A}{\vspace{0.1cm}\noindent(\textit{A.})}{}
\newenvironment{AC}{\vspace{0.1cm}\noindent(\textit{C and AC.})}{}
\newenvironment{runtime}{\vspace{0.1cm}\noindent(\textit{Running time.})}{}
\newcommand{\case}[1]{\item[\itshape\mdseries If {#1}:]}
\newenvironment{cs}{\begin{description}}{\end{description}}
\newcommand{\algorithmDAN}[1]{\vspace{0.1cm}\noindent
\underline{An algorithm for $#1$:}}

\definecolor{myBlue}{rgb}{0.0, 0.5, 1.0}
\definecolor{myRed}{rgb}{1.0, 0.5, 0}
\definecolor{myGreen}{rgb}{0.0, 0.5, 0.0}
\definecolor{myYellow}{rgb}{1.0, 1.0, 0.0}
\newcommand{\myRed}{myRed}
\newcommand{\myGreen}{myGreen}
\newcommand{\myBlue}{myBlue}

\renewcommand{\phi}{\varphi}
\renewcommand{\epsilon}{\varepsilon}

\newcommand{\NN}{{\mathbb N}}

\newcommand{\LA}{\textsf{\upshape A}}
\newcommand{\LB}{\textsf{\upshape B}}

\newcommand{\CA}{{\mathcal A}}
\newcommand{\CB}{{\mathcal B}}

\newcommand{\CJ}{{\mathcal J}}
\newcommand{\CK}{{\mathcal K}}
\newcommand{\CL}{{\mathcal L}}

\newcommand{\CO}{{\mathcal O}}

\newcommand{\CR}{{\mathcal R}}
\newcommand{\CS}{{\mathcal S}}

\newcommand{\CX}{{\mathcal X}}
\newcommand{\CY}{{\mathcal Y}}

\let\downarrowOld\downarrow
\renewcommand{\downarrow}{\mathord{\downarrowOld}}
\let\uparrowOld\uparrow
\renewcommand{\uparrow}{\mathord{\uparrowOld}}

\renewcommand{\hat}{\widehat}
\renewcommand{\tilde}{\widetilde}
\renewcommand{\hat}{\widehat}

\newcommand{\sym}{\operatorname{Sym}}

\newcommand{\aut}{\operatorname{Aut}}
\newcommand{\iso}{\operatorname{Iso}}
\newcommand{\stab}{\operatorname{Stab}}
\newcommand{\lab}{{\operatorname{Label}}}
\newcommand{\tranCl}{\operatorname{TClosure}}

\newcommand{\polylog}{\operatorname{polylog}}
\newcommand{\cw}{\operatorname{cw}}

\newcommand{\set}{{\operatorname{Set}}}
\newcommand{\intt}{{\operatorname{Int}}}
\newcommand{\setHyper}{{\operatorname{Hyper}}}
\newcommand{\obj}{\operatorname{Objects}}
\newcommand{\ord}{{\operatorname{Can}}}

\newcommand{\can}{\operatorname{CL}}
\newcommand{\canInt}{\operatorname{CL}_{\operatorname{Int}}}
\newcommand{\canSet}{\operatorname{CL}_{\operatorname{Set}}}
\newcommand{\canSetHyper}{\operatorname{CL}_{\operatorname{Hyper}}}
\newcommand{\canPoint}{\operatorname{CL}_{\operatorname{Point}}}
\newcommand{\canObj}{\operatorname{CL}_{\operatorname{Object}}}
\newcommand{\canMatch}{\operatorname{CL}_{\operatorname{Match}}}
\newcommand{\canGraph}{\operatorname{CL}_{\operatorname{Graph}}}
\newcommand{\cf}{\operatorname{CF}}

\newcommand{\enc}{\texttt{enc}}

%% file: articles/intro.tex
\section{Introduction}
The problem of computing a canonical form of
a graph can be seen as the task to compute a standard representative of the graph up to isomorphism. Specifically, given an input graph $G$, a graph $G'$ isomorphic to $G$ is to be computed such that the output graph $G'$ depends only on the isomorphism class of $G$ and not on the graph $G$ itself. The problem is closely related to the graph
isomorphism problem, which reduces to the task of computing a
canonical form:
for two given input graphs, we compute their
canonical forms and check whether the canonical forms are equal (rather than isomorphic).

In practice, a canonization algorithm is often preferable to an isomorphism
test, as it allows each graph to be treated separately, rather than having to compare
graphs pairwise. For example, when looking up a molecule in a chemical database, we do not wish to compare the molecule individually to every graph that has been stored in the system.

While various computational problems related to symmetries of graphs are polynomial-time equivalent
to the isomorphism problem, for example the
computation of the automorphism group of a graph or the computation of the orbit partition
\cite{DBLP:journals/ipl/Mathon79}, it is unknown whether or not
canonization reduces to isomorphism in polynomial time.

It is usually not very difficult to turn combinatorial isomorphism tests into canonization algorithms, sometimes the algorithms are canonization algorithms in the first place.
However, there are several isomorphism testing algorithms for which to date no canonization algorithms with the same asymptotic running are known.

This seems to be in particular the case for isomorphism algorithms based on group theoretic techniques pioneered by Luks \cite{DBLP:journals/jcss/Luks82}, who designed a polynomial time isomorphism algorithm for graphs of bounded degree.

Luks's framework and the result for bounded degree graphs were subsequently extended to
canonization~\cite{babai1983canonical}. In that paper, Babai and Luks lay the
foundation for canonization techniques using the string canonization problem and
algorithmically exploiting cosets. The original technique of Luks also sparked a
series of isomorphism algorithms without respective canonization algorithms. For
example, Luks presented a~$2^{O(|V|)}$-time isomorphism algorithm for
hypergraphs~\cite{DBLP:conf/stoc/Luks99}, but the best known canonization
algorithm has the brute force running time of~$|V|^{O(|V|)}$.
A similar situation occurs for hypergraphs of bounded rank \cite{DBLP:conf/focs/BabaiC08}.
As pointed out by Codenotti~\cite{codenotti2011testing} the reason is that these result are based on coset intersection, for which no adequate canonization version is known.
Finally, for Babai's quasipolynomial time isomorphism
algorithm~\cite{DBLP:conf/stoc/Babai16} there was initially no
canonization version. However, concurrently to our work, he extended his result
to canonization \cite{BabaiCanon}. For results building on his
algorithm~\cite{bounded-degree}, there are still no known canonization versions.
Also for the isomorphism problem for groups,
nearly all of the most recent
results seem not to provide canonical forms
\cite{DBLP:conf/icalp/BabaiCQ12,DBLP:journals/gcc/LewisW12,
DBLP:journals/tcs/RosenbaumW15,
DBLP:conf/isaac/GrochowQ15,
DBLP:journals/siamcomp/GrochowQ17,
Brooksbank17}.

\paragraph*{Our Contribution} In this paper, we devise a unified framework
for the design of canonization algorithms. To do so, we define, using hereditarily
finite sets, a general notion of combinatorial objects that includes graphs,
hypergraphs, relational structures, codes, permutation groups, tree
decompositions, and more.
To this end, we devise for various problems canonization algorithms with an asymptotic running time matching the best isomorphism algorithm. For each of them, no canonization algorithm with such a running time has been known before. Specifically, we obtain canonization algorithms for

\begin{itemize}
\item hypergraphs on a vertex set $V$ in time $2^{\CO(|V|)}$ matching the
hypergraph isomorphism algorithm of Luks~\cite{DBLP:conf/stoc/Luks99},
\item hypergraphs $X=(V,H)$ of color class size at most~$k$ in FPT
time~$2^{\CO(k)}(|V|+|H|)^{\CO(1)}$ matching the best isomorphism algorithm of
Arvind, Das, K{\"{o}}bler and Toda \cite{DBLP:journals/algorithmica/ArvindDKT15},

\item explicitly given permutation groups~$\Delta$ (up to permutational
isomorphism) on permutation domain $V$ and of order~$|\Delta|$ in
time~$2^{\CO(|V|)}|\Delta|^{\CO(1)}$ matching the best algorithm for
permutational isomorphism of Babai, Codenotti, and Qiao~\cite{DBLP:conf/icalp/BabaiCQ12},

\item explicitly given codes~$\CA$ (up to code equivalence) of code word
length~$|V|$ and code size~$|\CA|$ in time~$2^{\CO(|V|)}|\CA|^{\CO(1)}$ matching
the best isomorphism algorithm from
\cite{DBLP:conf/soda/BabaiCGQ11,DBLP:conf/icalp/BabaiCQ12}, and

\item combinatorial objects
over a ground set $V$ in general in time 
$2^{\CO(|V|)}n^{\CO(1)}$ where $n$ is the size of the object (formally defined
in Section~\ref{sec:comb:objs:and:lab:cos}).
\end{itemize}

The new canonization algorithm for hypergraphs above solves an open problem that has been repeatedly stated (\cite{DBLP:conf/stoc/Luks99},\cite{DBLP:conf/focs/BabaiC08}
and \cite{DBLP:journals/algorithmica/ArvindDKT15}).
As the input size of a
hypergraph can be as big as~$2^{|V|}$, when measured only in the size of the
underlying vertex set, the simply exponential
running time is optimal. Of course, when measured also in the number of edges,
better algorithms with a more refined running time bound might be possible.

For combinatorial objects in general, neither canonization nor isomorphism has
been considered before in this generality.

\paragraph*{Our Technique}
We advocate a clear separation
between unordered objects and ordered objects. In our framework, an unordered object is an object over a ground set of a priori indistinguishable vertices.
In contrast to this,
an ordered object is an object whose underlying vertex set consists of integers
which thus carries a natural linear order.
It is easily argued, that this induces a
polynomial-time computable linear ordering on the class of all ordered objects
(see \cref{lem:prec}). 

For canonization purposes, every unordered object is associated with an isomorphic ordered one, its canonical form, and the ordering of the ordered objects thereby extends to unordered objects. To compute the canonical form, we use the concept of a labeling coset, which is a coset of maps from unordered
to ordered objects. A labeling coset can be seen as a form of partial
canonization of an unordered object, through which various possible canonical forms have already been ruled out. By
representing labeling cosets compactly via generating sets, we can employ the
existing extensive library of efficient algorithms dealing with permutation
groups.

The fact that the ordered objects are totally ordered
allows us to use deterministic subroutines, used in isomorphism tests,
on ordered objects in an isomorphism-invariant way.
This framework then allows for a systematic transfer of the techniques that have
been developed for isomorphism testing to the realm of canonization. 

The most important
feature of our framework is that it is possible to view these labeling
cosets as combinatorial objects themselves.
The main technique (see \cref{lem:rep}) shows that under mild assumptions
it is possible to replace subobjects (such as subgraphs) by their labeling
cosets without losing or introducing global symmetries.
This paradigm allows for a recursive approach
by computing canonical labelings of substructures first,
then replacing the substructures by their labeling cosets,
and then computing a canonical labeling of a new global object that only consists of the labeling cosets of the substructures rather than the substructures themselves.

For hypergraphs, with a relatively direct application of our framework, we
obtain the canonization algorithm mentioned above. However, for the canonization
of objects in general the key algorithm handles sets of
labeling cosets.
As the labeling cosets can describe global interdependencies,
our technique of bundling and partitioning for the cosets is
significantly more involved and complicated (Section~\ref{sec:sets}).

We also
describe a technique of computing a canonical representation of ordered
groups which in turn allows us to associate a canonical string (encoding)
to every ordered object (see \cref{lem:enc}).
This complements our framework and allows us to use
arbitrary deterministic algorithms for ordered
objects as a black box in an isomorphism-invariant way.

Some of the techniques we describe here were used in a weaker, non-generic sense
in \cite{tree-width} to obtain the fastest known canonization algorithm
for graphs of bounded tree width.
In fact, tree
decompositions, nested tree decompositions, and treelike decompositions can also be viewed as combinatorial
objects in our framework.

\paragraph*{Organization of the Paper}
The goal of the paper is to compute
canonical labelings for arbitrary objects $\CX\in\obj(V)$ formally defined in
Section~\ref{sec:comb:objs:and:lab:cos}.
In a bootstrapping manner, we develop canonization algorithms for
objects that are more and more complex. Each algorithm uses the previous ones as a
black box.
In \Cref{sec:atoms}, we consider two easy
cases of canonization in which the object is a pair of atoms.
First, we show how to compute a canonical labeling
for $\CX$ in the case that
$\CX=(v,\Delta\rho)$ is a pair 
consisting of a vertex $v\in V$
and a labeling coset $\Delta\rho$.
Second, we show how to compute a canonical labeling for
$\CX$ in the case that
$\CX=(\Theta\tau,\Delta\rho)$ is a pair of two labeling cosets. 
In Section~\ref{sec:hyper}, we give the canonization algorithm for
hypergraphs, which uses the object replacement paradigm
explained in Section~\ref{sec:obj:replacement}.
In \Cref{sec:sets}, building on that we show how to compute canonical labelings
for $\CX$ in the case that $\CX=\{\Delta_1\rho_1,\ldots,\Delta_t\rho_t\}$
is a set of labeling cosets.
In \Cref{sec:obj}, we bring all our tools together to compute
canonical labelings for arbitrary objects $\CX\in\obj(V)$.
Finally, in Section~\ref{sec:Group}, we describe the technique of canonical
generating sets
which allows us to associate a canonical string (encoding)
to every ordered object
and which (in a certain sense) allows arbitrary deterministic group
theoretic algorithms to be used within our framework in an
isomorphism-invariant way.

%% file: articles/preliminaries.tex
\section{Preliminaries}

\paragraph{Set Theory}
For an integer $t$, we write $[t]$ for $\{1,\ldots,t\}$.
For a set $S$ and an integer $k$, we write $\binom{S}{k}$ for the $k$-element
subsets of $S$ and $2^S$ for the power set of $S$.

\paragraph{Group Theory}
We write composition of functions from left to right, e.g.,
for two functions $\phi:V\to V'$ and $\rho:V'\to V''$, we write
$\phi\rho$ for the function that first applies $\phi$ and then~$\rho$.
For a set $V$, we write $\sym(V)$ for the symmetric group on $V$
and for an integer $t$, we write $\sym(t)$ for $\sym([t])$.
By $\stab_\Psi(A):=\{\psi\in\Psi\mid \psi(a)\in A
\text{ for all }a\in A\}$, we denote the setwise
stabilizer of $A\subseteq V$ in $\Psi\leq\sym(V)$.
We sometimes also drop the index and write $\stab(A)$
for $\stab_{\sym(V)}(A)$.
For a vertex $v\in V$, we write
$\stab_\Psi(v)$ for $\stab_\Psi(\{v\})$.
We want to extend the definition to tuples $(A_1,\ldots,A_t)$
where $A_i\subseteq V$ and $t\geq 2$.
Inductively, we define $\stab_\Psi(A_1,\ldots,A_t)$ as
$\stab_{\stab_\Psi(A_1)}(A_2,\ldots,A_t)$.
Note that this way the
stabilizer
$\stab_\Psi(A_1,\ldots,A_t)=\{\psi\in\Psi\mid \psi(a)\in A_i
\text{ for all }i\in[t] \text{ and } a\in A_i\}$ can be computed
using an algorithm for the binary
stabilizer function $\stab_{-}(-)$
that gets an arbitrary permutation group $\Psi$
and only one set $A\subseteq V$ as an input.
A set $A\subseteq V$ is called $\Psi$-invariant
if $\stab_\Psi(A)=\Psi$.
For a permutation group $\Psi\leq\sym(V)$
and a vertex $v\in V$, we write $v^\Psi=\{\psi(v)\mid \psi\in\Psi\}$
for the $\Psi$-orbit of $v$.
The $\Psi$-orbit partition of $V$ is
a partition $V=V_1\cupdot\ldots\cupdot V_t$
such that $v,u\in V_i$ for some $i\in[t]$, if and only if
$v^\Psi=u^\Psi$.
A group $\Psi\leq\sym(V)$ is transitive on a~$\Psi$-invariant set $A\subseteq V$,
if $A$ consists of only one $\Psi$-orbit, i.e., $A=v^\Psi$
for some $v\in V$.
Slightly abusing terminology, a \emph{coset} of a set $V$ is a set~$\Lambda$ of
bijections from $V$ to a set~$V'$ such that~$\Lambda = \Delta
\rho=\{\delta\rho\mid\delta\in\Delta\}$ for some subgroup~$\Delta\leq\sym{(V)}$
and a bijection~$\rho:V\to V'$.

\paragraph{Generating Sets and Polynomial-Time Library}

For a set $S\subseteq\sym(V)$,
we write $\langle S\rangle$
for the smallest group $\Psi\leq\sym(V)$
for which $S\subseteq\Psi$.
In this case, $S$ is called a \emph{generating set}
for~$\Psi$. We refer to~\cite{seress} for the basic theory of handling
permutation groups algorithmically.
Many tasks can be performed efficiently
when a group is given implicitly via a generating set.
We list the results we use.

\begin{enumerate}
  \item Permutation groups and cosets
  can be
  represented implicitly via generating sets
  that can be chosen of size quadratic in $|V|$.
  \item The pointwise stabilizer $\stab_\Psi(v)$
  of a vertex $v\in V$ in a
  group $\Psi\leq\sym(V)$ can be computed
  with the Schreier-Sims algorithm in time
  polynomial in $|V|$.
  \item A subgroup of a permutation group with polynomial time membership problem
  can be computed in time polynomial
  in the index of the subgroup.
  \item Let $\CS=\Delta_1\rho_1,\ldots,\Delta_t\rho_t$ be
  a sequence of cosets of $V$. We write
  $\langle \CS\rangle$ for the smallest coset $\Lambda$
  such that $\Delta_i\rho_i\subseteq \Lambda$ for all $i\in[t]$.
  Given a representation for
  $\CS$, the coset $\langle \CS\rangle$ can
  be computed in polynomial time.
  Furthermore, the computation of $\langle \CS\rangle$ is
  isomorphism invariant w.r.t. $\CS$, i.e.,
  $\phi^{-1}\langle \CS\rangle=
\langle \phi^{-1}\CS \rangle$ for
all bijections $\phi:V\to V'$ (see~\cite[Lemma
9.1]{DBLP:journals/corr/GroheS15a}).
\end{enumerate}

%% file: articles/heredit.tex
\section{Combinatorial Objects and Labeling Cosets}\label{sec:comb:objs:and:lab:cos}

\paragraph{Labeling Cosets}
A \emph{labeling coset} of $V$ 
is set of bijective mappings
$\Delta\rho=\{\delta\rho\mid \delta\in\Delta\}$,
where $\rho$ is a
bijection from $V$ to $\{1,\ldots,|V|\}$
and $\Delta$ is a subgroup of $\sym(V)$.
Let $\rho:V\to\{1,\ldots,|V|\}$ be a bijection.
We write $\lab(V)$ for
the labeling coset $\sym(V)\rho=
\{\sigma\rho\mid \sigma\in\sym(V)\}$.
We say that $\Theta\tau$ is a \emph{labeling subcoset}
of a labeling coset $\Delta\rho$, written $\Theta\tau\leq\Delta\rho$,
if $\Theta\tau$ is a subset of $\Delta\rho$
and $\Theta\tau$ again forms a labeling coset.
For a labeling coset $\Delta\rho\leq\lab(V)$
and a $\Delta$-invariant set $A\subseteq V$, we define
the restriction of $\Delta\rho$ to $A$ as $(\Delta\rho)|_A
:=\{\lambda|_A\mid \lambda\in\Delta\rho\}$.
Observe that $(\Delta\rho)|_A$ is not necessarily a
labeling coset again since
the image of $\tilde\lambda\in(\Delta\rho)|_A$ might be a set
of natural numbers different from $\{1,\ldots,|A|\}$.
Let $\kappa$ be the unique bijection from
$\rho(A)$ to $\{1,\ldots,|A|\}$ that preserves the standard
ordering ``$<$'' of natural numbers.
We define the \emph{induced
labeling coset} of $\Delta\rho$ on $A\subseteq V$ as
$(\Delta\rho) \downarrow_A:=(\Delta\rho)|_A\kappa$.
Similarly, for a labeling coset $\Theta\tau\leq\lab(A)$,
we define the \emph{lifted
labeling coset} of $\Theta\tau$ to $V\supseteq A$ as
$(\Theta\tau)\uparrow^V:=\{\gamma\in\lab(V)\mid \gamma|_A\in\Theta\tau\}$.

\paragraph{Hereditarily Finite Sets and Combinatorial Objects}

We inductively define \emph{hereditarily finite sets}
over a ground set $V$.
Each vertex $X\in V$ and each
labeling coset $Y=\Delta\rho\leq\lab(V)$ is called an \emph{atom} and is in particular a hereditarily finite set.
Inductively, if $X_1,\ldots,X_t$ are hereditarily finite sets,
then $\CX=\{X_1,\ldots,X_t\}$ and also
$\CX=(X_1,\ldots,X_t)$ are hereditarily finite sets where $t\in\NN\cup\{0\}.$
A \emph{(combinatorial) object} is a pair~$(V,\CX)$ where~$\CX$ is a
hereditarily finite set over~$V$. The set of all (combinatorial) objects over~$V$ is denoted by $\obj(V)$. In the following, we will usually assume that~$V$ is apparent from context and identify~$\CX$ with the object~$(V,\CX)$, not distinguishing between hereditarily finite sets and combinatorial objects.

We say that an object is \emph{ordered} if the ground set $V$ is a linearly
ordered set.
An object is \emph{unordered}
if $V$ is an (unordered) set.
The linearly ordered ground sets considered in this paper are always
subsets of $\NN$ with their standard ordering ``$<$''. Additionally, we will never
consider partially ordered sets in which some but not all elements of $V$ are in
$\NN$.

The expressiveness of the object formalism is quite extensive.
In particular, we can view graphs $G=(V,E)$ with~$E\subseteq\binom{V}{2}$, hypergraphs
$X=(V,H)$ with~$H\subseteq 2^V$, and relational structures $Y=(V,R_1,\ldots,R_t)$ with~$R_i\subseteq V^{k_i}$
as unordered objects (over $V$).
A function $f:V\to V'$ can be encoded as a set
of pairs $\{(v,f(v))\mid v\in V\}$ and is thus an object.
Note that a labeling coset could in principle also be represented as a set of maps, and thus
as an object in which all atoms are of the type $X\in V$.
However, we want to succinctly represent labeling cosets via generating sets rather than as a set of labelings.
This is precisely the reason why we model them as a second kind of atom.

\paragraph{Applying Functions to Unordered Objects}
Let $V$ and $V'$ be ground sets where $V$ is unordered
and $V'$ is either unordered or ordered. 
Let $\mu:V\to V'$ be a bijection.
For an object $\CX\in\obj(V)$,
the image of $\CX$ under $\mu$,
written $\CX^\mu$, is an object over $V'$ defined as follows.
For a vertex $X=v\in V$, we define $X^\mu:=\mu(v)$.
For a labeling coset $Y=\Delta\rho\leq\lab(V)$, we
define $Y^\mu=(\Delta\rho)^\mu:=\mu^{-1}\Delta\rho$.
Inductively, 
for an object $\CX=\{X_1,\ldots,X_t\}$,
we define
$\CX^\mu:=\{X_1^\mu,\ldots,X_t^\mu\}$.
Analogously, for an object $\CX=
(X_1,\ldots,X_t)$, we define
$\CX^\mu:=(X_1^\mu,\ldots,X_t^\mu)$.
For example, the image of a graph $G=(V,E)$
under $\mu$
is a graph
$G^\mu=(V',E^\mu)$ that has
an edge $\{\mu(v),\mu(u)\}$,
if and only if $G$ has the edge $\{v,u\}$.

\paragraph{Automorphisms of Unordered Objects}
Two unordered objects $\CX\in\obj(V)$ and $\CX'\in\obj(V')$
are \emph{isomorphic}, written $\CX\cong\CX'$, if
there is a bijection $\phi:V\to V'$ such that
$\CX^\phi=\CX'$.
In this case, $\phi$ is called an \emph{isomorphism}.
The set of all isomorphisms 
from $\CX$ to $\CX'$ is denoted by $\iso(\CX;\CX')$.
Note that we defined
isomorphism only for unordered objects\footnote{
If we wanted to have a definition of isomorphisms for objects in
general,
we should require that an isomorphisms of an object must also preserve the
order of the corresponding ground set, which is consistent with the
framework.}.
The automorphism group of 
an unordered object $\CX$, written $\aut(\CX)$, consists of those
$\sigma\in\sym(V)$ for which $\CX^\sigma=\CX$. For unordered objects,
the automorphism group $\aut(\CX)$
often has a natural meaning,
e.g.,
$\aut((\Theta\tau,\Delta\rho))=\Theta\cap\Delta$ and
$\aut((A,\Delta\rho))=\stab_\Delta(A)$
where $A\subseteq V$ and
$\Theta\tau,\Delta\rho\leq\lab(V)$.

The set $\iso(V;V')$ simply consists of all
bijections from $V$ to $V'$.
However, instead of talking about bijections, we use the notation $\iso(V;V')$
to indicate and stress that both $V$ and $V'$ are unordered sets
and that for every object~$\CX$, every map~$\phi \in \iso(V;V')$ is an
isomorphism from~$\CX$ to $\CX^\phi$.

\paragraph{Canonical Forms of Unordered Objects}
A \emph{canonical form} is a function $\cf$ that
assigns each unordered object $\CX\in\obj(V)$
an ordered object $\cf(\CX)\in\obj(\{1,\ldots,|V|\})$ such that:

\begin{enumerate}[(\textnormal{CF}1)]
   \item\label{ax:cf2} $\CX\cong\CY$ implies
  $\cf(\CX)=\cf(\CY)$ for all objects $\CX,\CY$, and
  \item\label{ax:cf1} $\cf(\CX)=\CX^\lambda$
  for some $\lambda\in\lab(V)$.
\end{enumerate}

Condition (CF1) is called \emph{isomorphism invariance}
and is equivalent to $\cf(\CX)=\cf(\CX^\phi)$
for all $\phi\in\iso(V;V')$
\footnote{If we wanted to define an action of~$\phi$ on ordered objects, we
should define it to act trivially. In that case isomorphism invariance could
also be written as 
$\cf(\CX)^\phi=\cf(\CX^\phi)$, which is more consistent with definitions that
follow.}.
Intuitively,
Condition (CF2) means that $\CX$ and $\cf(\CX)$ are isomorphic
if we would forget about the linear ordering of the ground set of $\cf(\CX)$.

\paragraph{Canonical Labelings of Unordered Objects}
A \emph{canonical labeling function} $\can$ is
a function that assigns each unordered object
$\CX\in\obj(V)$ a labeling coset $\can(\CX)=\Lambda\leq\lab(V)$
such that:

\begin{enumerate}[(\textnormal{CL}1)]
  \item\label{ax:cl1} $\can(\CX)=\phi\can(\CX^\phi)$
  for all $\phi\in\iso(V;V')$ and,
  \item\label{ax:cl2} $\can(\CX)=\aut(\CX)\pi$
  for some (and thus for all) $\pi\in\can(\CX)$.
\end{enumerate}

Again, Condition (CL1)
is called \emph{isomorphism invariance}.
Note that (CL1) is equivalent to
$\can(\CX)^\phi=\can(\CX^\phi)$.
Roughly speaking, this means that $\can$
is compatible with isomorphisms.
More precisely, this means that the following diagram commutes.
\begin{center}
\begin{tikzcd}
\CX \arrow{r}{\phi} \arrow{d}{\can} & \CX^\phi \arrow{d}{\can} \\
\can(\CX)\arrow{r}{\phi} & \can(\CX)^\phi
\end{tikzcd}
\end{center}
Condition (CL2) means that $\CX^\lambda$
for an arbitrary labeling $\lambda\in\can(\CX)$
does not depend on the choice of $\lambda$.

We give a connection
between isomorphisms, canonical forms
and canonical labelings.
Deciding isomorphism reduces to computing a canonical form.
The computation of a canonical form in turn
reduces to computing canonical labelings as seen next.
We claim that
$\cf(\CX):=\CX^\lambda$, where $\lambda\in\can(\CX)$ is
an arbitrary labeling, defines a canonical
form.
First, observe that the canonical form
is well defined and
does not depend on the choice of $\lambda\in\can(\CX)$
since (CL2) holds.
Next, we check (CF1) and (CF2).
Condition (CF2) holds by definition.
For (CF1) we need to show that
$\cf(\CX)=\cf(\CX^\phi)$
for all $\phi\in\iso(V;V')$.
Because of (CL1), we have
that for all $\lambda\in\can(\CX)$
there is a $\lambda'\in\can(\CX^\phi)$
such that $\lambda=\phi\lambda'$.
Therefore, $\cf(\CX)=X^\lambda=(\CX^\phi)^{\lambda'}=\cf(\CX^\phi)$ which was to
show.

The notion of canonical forms and canonical labelings can
be defined naturally also for an isomorphism-closed class of objects (e.g., all hypergraphs).
In a bootstrapping manner, we will devise canonical labeling algorithms for
certain classes of combinatorial objects
until we finally give an algorithm for objects in general.
We will in each instance state specifically what the requirements of (CL1) and
(CL2) are.\footnote{We remark that in some papers on canonization there is a
Condition~(CL3) which we do not require in our framework as we
see labeling cosets as objects themselves.}

\paragraph{Representation of Objects}
For an object $\CX\in\obj(V)$, we define 
the \emph{transitive closure} of~$\CX$, written $\tranCl(\CX)$,
as all objects recursively occurring in $\CX$, i.e.,
$\tranCl(X):=\{X\}$ for $X=v\in V$
or $X=\Delta\rho\leq\lab(V)$.
And we define 
$\tranCl(\CX):=\{\CX\}\cup \bigcup_{i\in[t]} \tranCl (X_i)$
for $\CX=\{X_1,\ldots,X_t\}$ or
$\CX=(X_1,\ldots,X_t)$.
As mentioned in the preliminaries, we represent labeling cosets
efficiently using generating sets.
An object itself can be efficiently represented
as colored directed acyclic graph over
the elements in its transitive closure.
With this representation,
the input size (i.e., encoding length)
of an object $\CX$ is polynomial in $|\tranCl(\CX)|+|V|+t_{\operatorname{max}}$
where $t_{\operatorname{max}}$ is the maximal length of a tuple appearing in
$\tranCl(\CX)$.

\paragraph{The Linear Ordering of Ordered Objects}
In contrast to
the previous paragraphs,
we will now consider ordered objects where
$V\subseteq\NN$.
Such objects appear in the following context.
When evaluating a canonical form, the resulting object is
ordered.
In order to compare canonical forms,
it will be an important task
to sort ordered objects in polynomial time.
To do so, we define a linear order ``$\prec$''
on objects over the natural numbers.
For two natural numbers (atoms) $X,Y\in\NN$,
we adapt the natural ordering, i.e.,
$X\prec Y$ if $X<Y$.
We inductively extend our definition.
For two sets $\CX=\{X_1,\ldots,X_s\}$ and
$\CY=\{Y_1,\ldots,Y_t\}$, we assume that the order ``$\prec$''
is already defined for the elements $X_i$ and $Y_j$.
Then we say $\CX\prec\CY$ if
$|\CX|< |\CY|$ or if $|\CX|=|\CY|$ and the smallest
element in $\CX\setminus\CY$ is smaller
than the smallest element in $\CY\setminus\CX$.
Let $\CX=(X_1,\ldots,X_s)$ and
$\CY=(Y_1,\ldots,Y_t)$
be two tuples for which ``$\prec$'' is already defined
for the entries. We say $\CX\prec\CY$
if $s$ is smaller than~$t$
or if~$s=t$ and for the smallest $i\in[t]$ for which
$X_i\neq Y_i$, we have that
$X_i\prec Y_i$.
We extend the order to labelings
over natural numbers.
For two permutations $\sigma_1,\sigma_2\in\sym(\{1,\ldots,|V|\})$
we say $\sigma_1\prec \sigma_2$ if there
is an $i\in \{1,\ldots,|V|\}$ such that $\sigma_1(i) <\sigma_2(i)$
and $\sigma_1(j)=\sigma_2(j)$ for all $1\leq j<i$.
Last but not least, we extend the definition to
labeling cosets $\Delta\rho,\Theta\tau\leq\lab(\{1,\ldots,|V|\})$.
We adapt the definition for sets, i.e.,
$\Delta\rho\prec\Theta\tau$ if
$|\Delta\rho|\leq|\Theta\tau|$ or if
$|\Delta\rho|=|\Theta\tau|$ and
the smallest element of $\Delta\rho\setminus \Theta\tau$
is smaller than the smallest element
of $\Theta\tau \setminus \Delta\rho$.
It is known that for
two labeling cosets $\Delta\rho,\Theta\tau\leq\lab(\{1,\ldots,|V|\})$
given by generating sets, the order
``$\prec$'' can be computed in time polynomial in $|V|$
(\cite{tree-width}, Corollary 22).
For completeness, we define
$X\prec Y\prec \CX\prec \CY$ for all integers $X\in\NN$,
all labeling cosets $Y=\Delta\rho\leq\lab(\{1,\ldots,|V|\})$,
all tuples $\CX$ and all sets~$\CY$. 
We write $\CX\preceq\CY$ if $\CX=\CY$ or $\CX\prec\CY$.

\begin{lem}\label{lem:prec}
The ordering ``$\prec$'' on pairs of ordered objects can be computed in polynomial time.
\end{lem}

Having defined an ordering for ordered objects
that can be efficiently evaluated, we can use it as follows.
While we may not be able to distinguish non-isomorphic objects $\CX$ and $\CY$ per se, whenever we are given labelings $\lambda\in\can(\CX)$
and $\gamma\in\can(\CY)$, we can order the objects by
ordering the ordered versions $\CX^\lambda$ and $\CY^{\gamma}$
w.r.t. ``$\prec$''.

%% file: articles/cosetInt.tex
\section{Canonization of Atoms and Tuples}\label{sec:atoms}

We are looking
for a canonical labeling function
for objects $\CX$
in the case where the object is a pair
of atoms $\CX=(v,\Delta\rho)$ consisting of a labeling coset and a distinguished vertex.

\problem{\canPoint\label{prob:CL:point}}
{$(v,\Delta\rho)\in\obj(V)$ where $v\in V$, $\Delta\rho\leq\lab(V)$ and $V$ is
an unordered set}
{(v,\Delta\rho)}
{(\phi(v),\phi^{-1}\Delta\rho)}
{\stab_\Delta(v)}

The reader may want to take a moment to convince themselves that 
for input objects $\CX=(v,\Delta\rho)$, the
conditions
(CL1) and
(CL2) stated here agree with the condition~(CL1) and (CL2)
described in the previous section.

\begin{lem}\label{lem:canPoint}
A function $\canPoint$ solving Problem~\ref{prob:CL:point} can be computed in polynomial time.
\end{lem}

\begin{proof}
\algorithmDAN{\canPoint(v,\Delta\rho)}
\begin{cs}
\item~\vspace{-0.4cm}\\
Choose \emph{(arbitrarily)} $\rho^*\in\Delta\rho$ such that
$v^{\rho^*}\in\NN$ is minimal.\\
Set $v^\ord:=v^{\rho^*}$.\\
\emph{($v^\ord$ is a minimal image of $v$ under $\Delta\rho$.)}\\
Set $\Delta^\ord:=(\Delta\rho)^\rho$.\\
Return $\Lambda:=\rho^*\stab_{\Delta^\ord}(v^\ord)$.
\end{cs}

\begin{cl1}
Assume for some bijection~$\phi\in\iso(V;V')$ that we have
$\phi(v),\phi^{-1}\Delta\rho$ instead of $v,\Delta\rho$ as an input
of the algorithm.
Observe that $\phi^{-1}\rho$ is a coset representative
of $\phi^{-1}\Delta\rho=\phi^{-1}\Delta\phi(\phi^{-1}\rho)$.
We argue that the algorithm outputs $\phi^{-1}\Lambda$
instead of $\Lambda$.
Since $v^{\rho^*}=\phi(v)^{\phi^{-1}\rho^*}$,
we now obtain $\phi^{-1}\rho^*\delta^\ord$ instead of $\rho^*$
where $\delta^\ord$ is some element in $\stab_{\Delta^\ord}(v^\ord)$
(the choices for~$\rho^*$ vary up to elements in
$\stab_{\Delta^\ord}(v^\ord)$).
The computed objects~$ v^\ord=
\phi(v)^{\phi^{-1}\rho^*\delta^\ord}$ and $\Delta^\ord
=(\phi^{-1}\Delta\rho)^{\phi^{-1}\rho}$ remain unchanged.
The computed group $\stab_{\Delta^\ord}(v^\ord)$ remains unchanged.
Finally, the algorithm returns $\phi^{-1}\Lambda$
instead of $\Lambda$.
This gives $\canPoint(\phi(v),\phi^{-1}\Delta\rho)=
\phi^{-1}\Lambda=\phi^{-1}\canPoint(v,\Delta\rho)$
which was to show.
\end{cl1}

\begin{cl2}
This property holds by construction
since~$\rho^*\stab_{\Delta^\ord}(v^\ord)= \stab_\Delta(v)\rho^*$.
\end{cl2}

\begin{runtime}
The pointwise stabilizer $\stab_{\Delta^\ord}(v^\ord)$
can be computed with the
Schreier-Sims algorithm in polynomial time.
\end{runtime}
\end{proof}

We want to point out here that the Schreier-Sims-algorithm is applied to ordered objects only, and thus we do not need to worry about canonicity of its output.

We continue with the canonization of more interesting objects which we later use as subroutine
to canonize pairs of labeling cosets.
A (partial) \emph{matching} is a set $M\subseteq V_1\times V_2$
such that
for all $(v_1,v_2),(u_1,u_2)\in M$
with $(v_1,v_2)\neq(u_1,u_2)$, it holds that
$v_1\neq u_1$ and $v_2\neq u_2$.

\problem{\canMatch\label{prob:CL:Match}}
{$(M,\Delta\rho)\in\obj(V)$ where $M\subseteq V_1\times
V_2$ is a matching, $\Delta\rho\leq\lab(V)$,
$\Delta\leq\stab(V_1,V_2)$ and $V=V_1\cupdot V_2$
is an unordered set}
{(M,\Delta\rho)}
{(M^\phi,\phi^{-1}\Delta\rho)}
{(\aut(M)\cap\Delta)}

The reader may again want to take a moment to convince themselves that 
for input objects $\CX=(M,\Delta\rho)$, the
Conditions
(CL1) and
(CL2) stated here agree with Condition~(CL1) and (CL2)
described in the previous section.

\begin{lem}\label{lem:canMatch}
A function $\canMatch$ solving Problem~\ref{prob:CL:Match} can be computed in time
$2^{\CO(k_2)}|V|^{\CO(1)}$
where $k_2$ is the size of the largest $\Delta$-orbit of $V_2\subseteq V$.
\end{lem}

\begin{proof}
For the purpose of recursion, we use an additional input
parameter. Specifically, we use
a subset $A\subseteq V_2$ such that $M\subseteq V_1\times A$ and
$\Delta\leq\stab(V_1,A)$. Initially, we set $A=V_2$.

\algorithmDAN{\canMatch(M,A,\Delta\rho)}
\begin{cs}
\case{$|M|= 0$}
Return $\Lambda:=\Delta\rho$.
\case{$|A|= 1$}
\emph{(Because of $|M|\geq 1$, it holds that $|M|= 1$.)}\\
Assume $M=\{(v_1,v_2)\}$.\\
Return
$\Lambda:=\canPoint(v_1,\Delta\rho)$.
\case{$\Delta$ is intransitive on $A$}~\\
Partition $A=A_1\cupdot A_2$ where
$A_1$ is a $\Delta$-orbit such that
$A_1^\rho$ is minimal.\\
\minimal\\
Partition $M=M_1\cupdot M_2$
where $M_i:=\{(v_1,v_2)\in M\mid v_2\in A_i\}$ for $i=1,2$.\\
Compute
$\Lambda_1:=\canMatch(M_1,A_1,\Delta\rho)$ recursively.\\
Recurse and return
$\Lambda_2:=\canMatch(M_2,A_2,\Lambda_1)$.
\case{$\Delta$ is transitive on $A$}~\\
\transitive{\canMatch}{M}
\end{cs}

\begin{cl1}
Assume we have $M^\phi,A^\phi,\phi^{-1}\Delta\rho$
instead of $M,A,\Delta\rho$ as an input.
Observe that $\phi^{-1}\rho$ is a coset representative
for $\phi^{-1}\Delta\rho$.
We show that the algorithm outputs $\phi^{-1}\Lambda$
instead of $\Lambda$.

In the base case $|M|=0$, we
return $\phi^{-1}\Delta\rho$ instead of $\Delta\rho$.

Now, consider the case $|A|=1$.
We obtain $\phi(v_1)$ and $\phi(v_2)$
instead of $v_1$ and $v_2$, respectively.
By (CL1) of $\canPoint$, we return $\phi^{-1}\Lambda$
instead of $\Lambda$.

In the intransitive case, we obtain
$A^\phi=A_1^\phi\cupdot A_2^\phi$ as a partition
since $A_1^\rho=A_1^{\phi\phi^{-1}\rho}$.
By induction, we obtain $\phi^{-1}\Lambda_1$
instead of $\Lambda_1$.
Again, by induction, we return
$\phi^{-1}\Lambda_2$
instead of~$\Lambda_2$.

In the transitive case, observe that $A^\ord$ and $\Delta^\ord$
remain unchanged since
$A^{\phi\phi^{-1}\rho}=A^{\rho}$
and $(\phi^{-1}\Delta\rho)^{\phi^{-1}\rho}=(\Delta\rho)^{\rho}$.
Therefore, we still obtain $A^\ord=A^\ord_1\cupdot A^\ord_2$
and also $\Psi^\ord$.
We obtain cosets $\phi^{-1}\rho\delta_i^\ord\Psi^\ord$
instead of $\rho\delta^\ord_j\Psi^\ord$ since
the indexing is arbitrary.
The calls we do now are
of the form
$\canMatch(M^\phi,A^\phi,\phi^{-1}\delta^\ord_i\Psi^\ord)$
instead of $\canMatch(M,A,\delta_j^\ord\Psi^\ord)$.
By induction, we get $\phi^{-1}\Delta_i\rho_i$
instead of $\Delta_j\rho_j$.
We obtain $\phi^{-1}\rho_i$
instead of $\rho_j$.
But since $(M,\Delta\rho)^{\rho_i}
=(M^\phi,\phi^{-1}\Delta\rho)^{\phi^{-1}\rho_i}$
we get the same ordered sequence.
The computation of $\Lambda$ is isomorphism invariant
and
therefore we return $\phi^{-1}\Lambda$
instead of $\Lambda$.
\end{cl1}

\begin{cl2}
In the Case that $M=\emptyset$,
it holds $\aut(M)\cap\Delta=\Delta$.

In Case $|A|=1$ and thus $|M|= 1$, it holds that
$\aut(M)\cap\Delta$
is actually equal to $\stab_\Delta(v_1)$ since
$\{v_2\}=A$ is already stabilized by $\Delta$.

For the intransitive case, we have $\Lambda_1=
(\aut(M_1)\cap\Delta)\pi_1$ for some $\pi_1\in\Lambda_1$ by induction.
Again, by induction, the returned
coset is $\Lambda_2=(\aut(M_2)\cap\aut(M_1)\cap\Delta)\pi_2$
for some $\pi_2\in\Lambda_2$.
Since $(M_1,M_2)$ is a $\Delta$-invariant
\emph{ordered} partition of $M$,
it holds that $\Lambda_2=(\aut(M)\cap\Delta)\pi_2$.

For the transitive case,
observe that (CL1)
of $\canMatch$
implies
$(\aut(M)\cap\Delta)\pi\subseteq\Lambda$
for some $\pi\in\Lambda$.
Next, we show the reversed inclusion
$\Lambda\subseteq(\aut(M)\cap\Delta)\pi$.
We need to show that  $\rho_i\rho_j^{-1}$
is an element in $\aut(M)\cap\Delta$ for all $i,j\in[r]$.
The membership $\rho_i\rho_j^{-1}\in\aut(M)$
follows from the fact that $M^{\rho_i}=M^{\rho_j}$
and the membership $\rho_i\rho_j^{-1}\in\Delta$ follows
from the fact that $(\Delta\rho)^{\rho_i}=(\Delta\rho)^{\rho_j}$.

\end{cl2}

\begin{runtime}
Let $A^*\subseteq A\subseteq V_2$ be a $\Delta$-orbit of maximal size.
We claim that the maximum number of recursive calls~$R(|A^*|,|A|)$ is at most~$
T:=2^{6|A^*|}|A|^2$.
In the intransitive case, this is easy to see by induction:
\begin{equation*}
R(|A^*|,|A|)\leq 1+\sum_{j\in[2]} R(|A^*|,|A_j|)
\overset{\text{induction}}{\leq}
1+2^{6|A^*|}(|A_1|^2+|A_2|^2)\leq T.
\end{equation*}
In the transitive case, it holds that $A^*=A$
and $s\leq 2^{|A|}$ and we obtain
\begin{equation*}
R(|A|,|A|)\leq 1+ s\cdot R(\lceil|A|/2\rceil,|A|)
\overset{\text{induction}}{\leq}
1+2^{4|A|+3}|A|^2
\overset{2\leq |A|}{\leq}
T.
\end{equation*}
We consider the running time for one single call without recursive costs. All steps are
polynomial time computable, except the computation of $\Psi^\ord$ in the transitive case.
The group $\Psi^\ord$ can be computed
in time polynomial in the index and $|V|$,
i.e., ${2^{\CO(|A^*|)}|V|^{\CO(1)}}$.
In total, we have a running time of at most
$T\cdot 2^{\CO(|A^*|)}|V|^{\CO(1)}
\subseteq 2^{\CO(k_2)}|V|^{\CO(1)}$
where $k_2$ is the size of the largest $\Delta$-orbit of $A\subseteq V_2$.
\end{runtime}
\end{proof}

As next step, we demonstrate how to compute canonical labelings
for objects $\CX=(\Theta\tau,\Delta\rho)$ consisting of a pair
of labeling cosets.

\problem{\canInt\label{prob:CL:Int}}
{$(\Theta\tau,\Delta\rho)\in\obj(V)$ where $\Theta\tau,\Delta\rho\leq\lab(V)$
and $V$ is an unordered set}
{(\Theta\tau,\Delta\rho)}
{(\phi^{-1}\Theta\tau,\phi^{-1}\Delta\rho)}
{(\Theta\cap\Delta)}

In fact, Condition (CL2) given here coincides with the general
Condition (CL2) stated in the preliminaries, i.e.,
$\canInt(\Theta\tau,\Delta\rho)=\aut((\Theta\tau,\Delta\rho))\pi$ for some $\pi$.
The problem can be seen as a canonization analogue to
the group-intersection problem which is often used
for the purpose of designing isomorphisms algorithms.
The next example shows how the problem can be
used for canonization purposes.

\begin{exa}\label{exa:int}
Let $G=(V,E,P)$ be a graph with
an \emph{ordered} partition $P=(E_B,E_R)$
of the edges $E=E_B\cupdot E_R$.
The partition can be seen as an edge coloring
which has to be preserved by automorphisms
of the graph
(see
Figure~\ref{fig:distinguishable:colors}).
\begin{figure}[t]
\begin{center}
\begin{tikzpicture}[scale=0.7]
\foreach \i in {1,...,6}{
	\node[draw,circle,scale=0.5,color=black]
	(\i) at (360/6*\i:1cm)
	{};}
\path[draw,\myRed,ultra thick]
(1) -- (2) -- (3) -- (4) -- (5) -- (6) -- (1);
\path[draw,\myBlue,ultra thick]
(1) -- (3) -- (5) -- (1);
\end{tikzpicture}
\caption{A graph $G=(V,E,P)$ with an \emph{ordered} partition $P=(
\bm{\textcolor{\myBlue}{E_B}},
\bm{\textcolor{\myRed}{E_R}}
)$.\label{fig:distinguishable:colors}}
\end{center}
\end{figure}
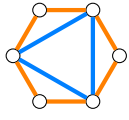
Let $\can,\canInt$ be canonical labeling functions
and define $\Delta_B\rho_B:=\can(E_B)$
and $\Delta_R\rho_R:=\can(E_R)$.
Then $\canGraph(G):=\canInt(\Delta_B\rho_B,\Delta_R\rho_R)$ defines
a canonical labeling for graphs with edges colored blue and red.
\end{exa}

\begin{lem}\label{lem:canInt}
A function $\canInt$ solving
\cref{prob:CL:Int}
can be computed in time
$2^{\CO(k)}|V|^{\CO(1)}$
where $k$ is the size of the largest $\Delta$-orbit of $V$.
\end{lem}

\begin{proof}
Let $\tilde V:=\{\tilde v_1,\ldots,\tilde v_{|V|}\}$
be a set of size $|V|$ disjoint from $V$.
The set $\tilde V$ is essentially a copy of~$V$.
Define $U:=\tilde V\cupdot V$.
Let $\Delta_U\rho_U\leq\lab(U)$ be the labeling coset
on $U$
obtained by
$\Theta\tau$ acting on $\tilde V$
and $\Delta\rho$ acting on $V$.
More formally, we define $\Delta_U\rho_U:=\{
\lambda_U\in\lab(U)\mid\exists
\gamma\in\Theta\tau,\lambda\in\Delta\rho:\text{ for all }
i,j\in \{1,\ldots,|V|\} \text{ we have }
\lambda_U(\tilde v_i)=\gamma(v_i)+|V|$
and 
$\lambda_U(v_j)=\lambda(v_j)\}$.
Define a matching $M:=\{(\tilde v_i,v_i)\mid i\in \{1,\ldots,|V|\}\}$
by pairing corresponding vertices.
Define $\Lambda_U:=\canMatch(M,\Delta_U\rho_U)$.
We claim that $\Lambda:=\Lambda_U\downarrow_V$
defines
a canonical labeling for $(\Theta\tau,\Delta\rho)$
where ``$\downarrow$''
denotes the induced labeling coset
(as defined at the beginning of Section~\ref{sec:comb:objs:and:lab:cos}).

\begin{cl1}
Assume we have $\phi^{-1}\Theta\tau,\phi^{-1}\Delta\rho$
instead of $\Theta\tau,\Delta\rho$ as an input.
Following the construction, we obtain $M^{\phi_U}$ instead of $M$
and $\phi_U^{-1}\Delta_U\rho_U$ instead of $\Delta_U\rho_U$
for some $\phi_U$ with $\phi_U|_V=\phi$.
By (CL1) of $\canMatch$, we obtain $\phi_U^{-1}\Lambda_U$
instead of $\Lambda_U$.
Therefore, we obtain $\phi^{-1}\Lambda$
instead of $\Lambda$.
\end{cl1}

\begin{cl2}
We need to show that $(\aut(M)\cap\Delta_U)|_V=\Theta\cap\Delta$.
The inclusion $\Theta\cap\Delta\subseteq(\aut(M)\cap\Delta_U)|_V$
follows from (CL1) of this reduction.
We thus need to show the reversed inclusion
$(\aut(M)\cap\Delta_U)|_V\subseteq\Theta\cap\Delta$.
So assume $\delta_U\in\aut(M)\cap\Delta_U$.
Let $f:\tilde V\to V$ be the bijection
such that $f(\tilde v_i)= v_i$.
Since $\delta_U\in\Delta_U$,
there are some $\theta\in\Theta,\delta\in\Delta$ such that
$f(\delta_U(\tilde v_i))=\theta(v_i)$
and $\delta_U(v_i)=\delta(v_i)$ for all $i\in \{1,\ldots,|V|\}$.
Since $\delta_U\in\aut(M)$, it holds that
$f(\delta_U(\tilde v_i))=\delta_U(f(\tilde v_i))=\delta_U(v_i)$
for all $i\in\{1,\ldots,|V|\}$.
Both together imply that $\theta(v_i)=f(\delta_U(\tilde v_i))=\delta_U(v_i)=\delta(v_i)$
for all $v_i\in V$.
Thus $\delta_U|_{V}\in\Theta\cap\Delta$.
\end{cl2}

\begin{runtime}
Observe that the exponential term in the running time of $\canMatch$ just depends
on the size of the largest $\Delta_U$-orbit of $V\subseteq U$
which is simply the size of the largest $\Delta$-orbit
of $V$.
\end{runtime}
\end{proof}

An algorithm similar to the one just described can be found~\cite{miller1983isomorphism}.
So far, we are able to canonize a pair of two atoms.
This is already sufficient to canonize tuples.

\paragraph{Canonization
of Tuples using Iterated Instances}
The three canonical labeling functions we have defined so far take two inputs,
namely some object as first input (a point/a matching/a labeling coset) and a labeling coset as second input.
We want to extend the definition
to more arguments replacing the first input by a tuple of objects. This is in analogy to our definition of the stabilizer
function $\stab_{-}(-)$.
Let $\can$ be a canonical labeling function
for pairs of objects $(\CX,\Delta\rho)$
where $\CX\in\obj(V)$ and $\Delta\rho\leq\lab(V)$.
For an object $X_1$,
we define $\can(X_1;\Delta\rho)$
as $\can(X_1,\Delta\rho)$.
Inductively for $t\geq 2$, we define
the \emph{iterated instance}
$\can(X_1,\ldots,X_t;\Delta\rho)$
as $\can(X_2,\ldots,X_t;\can(X_1,\Delta\rho))$.
To justify the definition,
let
$X_1,\ldots,X_t$ be a sequence of objects
and let $\Delta\rho\leq\lab(V)$
be a labeling coset.
Then
$\can(X_1,\ldots,X_t;\Delta\rho)$ defines a
canonical labeling for
the tuple $(X_1,\ldots,X_t,\Delta\rho)$.
In order to compute canonical labelings
for tuples, it is therefore sufficient
to canonize pairs.

%% file: articles/objectReplacement.tex
\section{Object Replacement}\label{sec:obj:replacement}
A crucial advantage of the framework we have defined is
that labeling
cosets are objects themselves.
The next lemma shows that for a given set $\CX=\{X_1,\ldots,X_t\}$ of pairwise
isomorphic objects $X_i$, we can replace the objects $X_i$ by their
canonical labeling cosets without losing or introducing global symmetries of the
object $\CX$.
The benefit is that the canonical labelings of the objects $X_i$ can be
computed using a recursive approach.

\begin{lem}[Object replacement]\label{lem:rep}
Let $\CX=\{X_1,\ldots,X_t\}$ be an object and let
$\can,\canSet$ be canonical labeling functions and let
$\Delta_i\rho_i:=\can(X_i)$ for each $i\in[t]$ and
$\CX^\set:=\{\Delta_1\rho_1,\ldots,\Delta_t\rho_t\}$.
Assume that
$X_i^{\rho_i}=X_j^{\rho_j}$ for all $i,j\in [t]$.
Then $\canObj(\CX):=\canSet(\CX^\set)$ defines a canonical labeling
for $\CX$.
\end{lem}

\begin{proof}
We have $X_i^{\rho_i}
=X_j^{\rho_j}$ for all
$i,j\in[t]$.
Then, by (CL1) and (CL2)
of the function $\can$, it holds for all $\sigma\in\sym(V)$
that $X_i^{\sigma}=X_j$, if and only if
$\can(X_i^{\sigma})=\can(X_j)$, which in turn is equivalent
to $(\Delta_i\rho_i)^\sigma=\Delta_j\rho_j$.
This implies
that $\aut(\CX)=\aut(\CX^\set)$.
\end{proof}

We extend \cref{exa:int} to (unordered) partitions by exploiting the object
replacement lemma.

\begin{exa}\label{exa:set}
Let $G=(V,E,P)$ be a graph
with an (unordered) partition $P=\{E_B,E_R\}$
of the edges $E=E_B\cupdot E_R$.
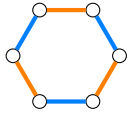
\begin{figure}[t]
\begin{center}
\begin{tikzpicture}[scale=0.7]
\foreach \i in {1,...,6}{
	\node[draw,circle,scale=0.5,color=black]
	(\i) at (360/6*\i:1cm)
	{};}
\path[draw,\myRed,ultra thick]
(1) -- (2);
\path[draw,\myRed,ultra thick]
(3) -- (4);
\path[draw,\myRed,ultra thick]
(5) -- (6);
\path[draw,\myBlue,ultra thick]
(2) -- (3);
\path[draw,\myBlue,ultra thick]
(4) -- (5);
\path[draw,\myBlue,ultra thick]
(6) -- (1);
\end{tikzpicture}
\caption{A graph $G=(V,E,P)$
with an (unordered) partition $P=\{
\bm{\textcolor{\myBlue}{E_B}},
\bm{\textcolor{\myRed}{E_R}}\}$.\label{fig:indistinguishable:colors}}
\end{center}
\end{figure}
In contrast to \cref{exa:int},
the partition is not ordered.
Here,
there exists
an automorphism that
maps the set of blue edges $E_B$ to the set of red edges $E_R$ (see Figure~\ref{fig:indistinguishable:colors}).
Let $\can,\canSet$ be canonical labeling functions
and let $\Delta_B\rho_B:=\can(E_B)$
and $\Delta_R\rho_R:=\can(E_R)$.
By \cref{lem:rep},
$\canGraph(G):=\canSet(\{\Delta_B\rho_B,\Delta_R\rho_R\})$
defines a canonical labeling for graphs with (unordered) edge partitions.
\end{exa}

To exploit the object replacement lemma, our upcoming algorithms in
\cref{sec:hyper,sec:sets}
use the following partitioning strategy.

\paragraph{General Strategy}
Assume we are given an object $X=\{x_1,\ldots,x_t\}$ that is to be canonized.
Our general strategy is to construct an (\emph{unordered}) partition $X=X_1\cupdot\ldots\cupdot
X_s$ in an isomorphism-invariant way. We call the parts~$X_i$ \emph{bundles}.
The isomorphism invariance of the partition into bundles in particular
ensures that the set $\{X_1,\ldots,X_s\}$ has the same
automorphism group as $X$.
In the cases where this leads to a trivial partition (i.e., $s=t$ or $s=1$),
we will use problem specific arguments to instantly make some form of progress.

The tough case will occur when the partition is non-trivial.
In this case, we will recursively compute
canonical labelings $\Theta_i\tau_i$ for the bundles $X_i$.
(Since $s>1$ and thus $X_i\subsetneq X$, we are guaranteed to make progress  on the
recursive instance).
This gives us canonical labelings for each bundle independently,
but we have not taken into account any interdependencies between the bundles. This will be our next step.

First, assume that the bundles $X_i$ are pairwise non-isomorphic.
Since each bundle is canonized separately,
we are now able to sort the bundles according to their canonical forms.
This gives an \emph{ordered} partition of the bundles
$\CX:=(X_1,\ldots,X_s)$
that is obtained
by renaming the indices such that
$X_i^{\tau_i}
\prec X_j^{\tau_j}$,
if and only if
$i<j\in[s]$.
Now, we can exploit the object replacement paradigm
and replace the bundles
by their labeling cosets without losing any information.
This leads to a tuple of
labeling cosets $\CX^\intt:=(\Theta_1\tau_1,\ldots,\Theta_s\tau_s)$ for which we
already have a canonization technique in \cref{sec:atoms}.

Second, assume the bundles $X_i$ are pairwise
isomorphic.
Here, we have to deal with the (unordered) partition
$\CX:=\{X_1,\ldots,X_s\}$.
We exploit the object replacement lemma (\cref{lem:rep}) again.
We replace each bundle with its labeling coset and
obtain a set of labeling cosets
$\CX^\set:=\{\Theta_1\tau_1,\ldots,\Theta_s\tau_s\}$.
Doing this, we do not lose any symmetries of the object $\CX$
and it is now sufficient to compute a canonical labeling
for the object $\CX^\set$ rather than $\CX$.
To canonize $\CX^\set$, we will use a recursive approach. (Since $s<t$ and thus
$|\CX^\set|<|X|$, we are guaranteed to make progress on the recursive instance).

In general, we would see a mixture of the cases
with some bundles being isomorphic to others but not to all. In this mixed case,
we order the bundles according to their isomorphism type and perform a mixture of the other two cases.

%% file: articles/hypergraphCanonizationExp.tex
\section{Canonization of Hypergraphs}\label{sec:hyper}

As dictated by the object replacement lemma (\cref{lem:rep}), the key problem we need to
solve is to compute a canonical labeling function for objects
$\CX^\set$ where the object $\CX^\set=\{\Delta_1\rho_1,\ldots,\Delta_t\rho_t\}$ 
is a set of labeling cosets.
Towards a solution of this, our next building block  will be the following more specialized problem, reminiscent of hypergraph canonization.

\problem{\canSetHyper\label{prob:CL:SetHyper}}
{$(K,\Delta\rho)\in\obj(V)$ where
$K=\{(\Delta_1\rho_1,S_1),\ldots,(\Delta_t\rho_t,S_t)\}$,
$\Delta_i\rho_i\leq\lab(V)$, $\Delta\rho\leq\lab(V)$, $S_i\subseteq V$ for
all $i\in[t]$, $S_i\neq
S_j$ for $i\neq j$ and $V$ is an unordered set}
{(K,\Delta\rho)}
{(K^\phi,\phi^{-1}\Delta\rho)}
{\{\delta\in\Delta\mid\exists\psi\in\sym(t)
\forall
i\in[t]:(\Delta_i\rho_i,S_i)^\delta
=(\Delta_{\psi(i)}\rho_{\psi(i)},S_{\psi(i)})\}}

As usual, Conditions (CL1) and (CL2) given here coincide with the general
Condition (CL1) and (CL2). In particular,
$\canSetHyper(K,\Delta\rho)=\aut((K,\Delta\rho))\pi$ for some $\pi$.

\begin{lem}\label{lem:canSetHyper}
A function $\canSetHyper$
solving \cref{prob:CL:SetHyper}
can be computed in time
$2^{\CO(k)}n^{\CO(1)}$
where~$n$ is the input size and
$k$ is the size of the largest $\Delta$-orbit of $V$.
\end{lem}

In an instance of Problem~\ref{prob:CL:SetHyper}, each labeling coset~$\Delta_i\rho_i$ comes with a subset $S_i$ of~$V$ and these subsets are pairwise distinct.
In the special case where $\Delta_i\rho_i=\lab(V)$ for all $i\in[t]$, the problem is
equivalent to hypergraph canonization.
However, we need the more general version for a recursive approach.
Nevertheless, we think of the sets~$S_i$ as hyperedges. Before giving a detailed
proof we describe the general strategy.

\paragraph{Intuition for the Hypergraph Algorithm}

To solve Problem~\ref{prob:CL:SetHyper}, we will maintain at any point in time a~$\Delta$-invariant
set $A\subseteq V$ for which the following condition (Condition (A)) holds:
$S_i\cap A\neq
S_j\cap A$ for all $i\neq j$. 
This set $A$ measures our progress in the sense that the
number of hyperedges
$t$ is bounded by $2^{|A|}$.
In the base case, if~$|A|$ is smaller than an absolute constant, we have a constant
number of hyperedges and can apply a brute force algorithm.

If~$\Delta$ acts transitively on~$A$ (the transitive case),
we branch on all possible splits of $A$ into two equally sized sets $A_1,A_2$,
shrinking~$\Delta$ in the process, and achieving that~$\Delta$ is not transitive anymore.

The trickiest case is when~$\Delta$ does not act transitively on~$A$. This gives us a canonical~$\Delta$-invariant partition 
$A=A_1\cupdot A_2$ (the intransitive case). Here, we need to apply recursion.
A primitive approach by recursing on $A_1$ and subsequently on $A_2$ does not work,
since it would not maintain the Condition~(A) which ensures
that we can handle the base case.
Instead, we define a partition into bundles $K=K_1\cupdot\ldots\cupdot K_s$
of the set $K$ as proposed by our general strategy.
In the concrete situation for hypergraphs, we can define bundles as follows.
We say that two hyperedges $S_i$ and $S_j$
are in the same bundle, if and only if
the hyperedges agree in the first part~$A_1$ (i.e.,
$S_i\cap A_1=S_j\cap A_1$)
(see \cref{fig:hyper}).
This gives an (unordered) partition of the set of
hyperedges $K=K_1\cupdot\ldots\cupdot K_s$, where each subset of hyperedges $K_i$ corresponds to a bundle.
In the extreme case where each hyperedge forms its own bundle, and hence $s=t$,
we can restrict our focus $A$ to the set $A_1$ and we still maintain Condition
(A), making progress. In the other extreme, if all hyperedges correspond to the
same bundle, and hence $s=1$, we can restrict our set $A$ to the $A_2$.
It remains to explain the case in which we have a proper bundling.
For this observe that the bundles themselves form instances of our original problem so
we can compute a canonical labeling
$\Theta_i\tau_i:=\canSetHyper(K_i)$ for each of the bundles recursively.
Following our general strategy, we next compute a canonical labeling for
the set of bundles taking
the relation between the bundles into account.

For the case that the bundles $K_\ell$ are pairwise non-isomorphic,
we directly follow our general strategy.
We order the bundles according to their isomorphism type and canonize them
sequentially using iterated instances from \cref{sec:atoms}.

Assume now that the bundles $K_i$ are pairwise
isomorphic.
In this case, we have an (\emph{unordered}) partition $\CK:=\{K_1,\ldots,K_s\}$
consisting of the pairwise isomorphic bundles.
Here, we exploit the object replacement
paradigm (\cref{lem:rep}).
We replace each bundle with its labeling coset and
obtain a set of labeling cosets
$\CK^\set:=\{\Theta_1\tau_1,\ldots,\Theta_s\tau_s\}$.
Doing this, we do not lose any symmetries of the hypergraph $\CK$
and it is now sufficient to compute a canonical labeling
of the object $\CK^\set$ rather than $\CK$.
To canonize $\CK^\set$ using recursion, we have to interpret this set
of labeling cosets as an
instance of our hypergraph problem.
For this, observe that restricted to the set $A_1\subseteq A\subseteq V$
the set $\CK^\set$ must induce a hypergraph structure since the bundles 
pairwise disagree on~$A_1$
(see the left side of \cref{fig:hyper}).
This allows us to compute the canonical labeling for $\CK^\set$ recursively
using our hypergraph algorithm.

In the general case, in which some bundles being isomorphic to others but not to
all, we perform a mixture of the two cases just described.
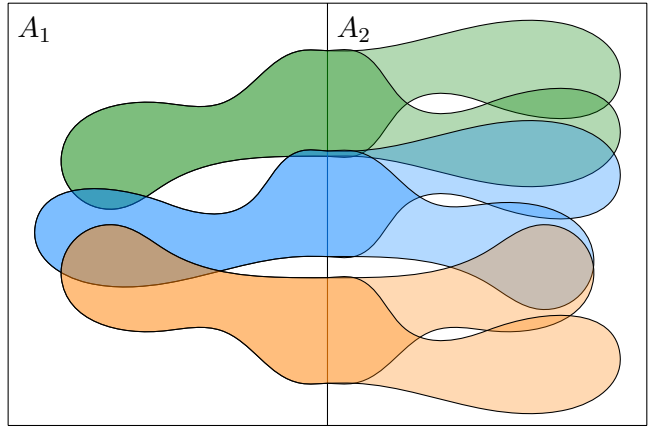
\begin{figure}[t]
\begin{center}
\input{./figures/bundling.tex}
\caption{Bundling the hyperedges: six hyperedges are bundled into 3 bundles (shown in distinct colors) each consisting of two hyperedges.\label{fig:hyper}}
\end{center}
\end{figure}

\paragraph{Comparison to Previous Algorithms}
The first $2^{\CO(|V|)}$-time hypergraph-isomorphism algorithm is due to Luks
\cite{DBLP:conf/stoc/Luks99}.
In a dynamic programming fashion,
he computed the isomorphisms by intersecting
cosets of isomorphisms between already computed subhypergraphs.
This approach makes use of the coset-intersection problem
for permutation groups. 

Using Luks's approach as a starting point, in \cite{DBLP:journals/algorithmica/ArvindDKT15} an algorithm is developed for the setting of bounded color classes.
Indeed, the authors found a way to
exploit ordered partitions by using the bundling technique described here.
However, within the color classes they used essentially Luks's dynamic programming algorithm
to compute the isomorphisms.

The crucial novelty in our algorithm is the more general definition
of the problem in which labeling cosets occur as part of the input structure.
This allows us to completely give up dynamic programming,
and instead use the bundling technique in a recursive way.
Treating cosets (which are the results of recursive canonization calls) as objects themselves allows us to recurse on substructures.

This in turn allows us to compute
canonical labelings for hypergraphs resolving the open question of both papers.

Another advantage of our approach is that it allows for a polynomial-time
algorithm if the group $\Delta$ is of a certain restricted structure (such as having bounded size composition factors, see \cref{cor:gammaD}).

\paragraph{Detailed Description of the Hypergraph Algorithm}
We proceed to prove Lemma~\ref{lem:canSetHyper} by giving a detailed description and analysis of our hypergraph  canonization algorithm.

\begin{proof}[Proof of Lemma~\ref{lem:canSetHyper}]
For the purpose of recursion, we need an additional input parameter.
Specifically, we use a subset $A\subseteq V$ such that
$\Delta\leq\stab(A)$. 
Initially, we set $A=V$.
For this parameter, we always require that
\begin{enumerate}
  \item[\textnormal{(A)}] $S_i\cap A\neq
S_j\cap A$ for all $i\neq j$.
\end{enumerate}

\algorithmDAN{\canSetHyper(K,A,\Delta\rho)}
\begin{cs}

\case{$|A|\leq 1$}
\emph{(Because of Property (A), it holds that $|K|\leq 2$.)}
\begin{cs}
\case{$|K|=1$}
Return $\canInt(\Delta_1\rho_1,\lab(S_1)\uparrow^V;\Delta\rho)$.\\
\iterated
\case{$|K|=2$}~\\
Rename the indices of elements in $K$ such that
$S_1\cap A=\emptyset$
and $S_2\cap A=A$.\\
Return
$\canInt(\Delta_1\rho_1,\lab(S_1)\uparrow^V,\Delta_2\rho_2,\lab(S_2)\uparrow^V;\Delta\rho)$.\\
\iterated
\end{cs}

\case{$\Delta$ is intransitive on $A$}~\\
Partition $A=A_1\cupdot A_2$ where
$A_1$ is a $\Delta$-orbit such that
$A_1^\rho$ is minimal.\\
\minimal\\
Define $S_{i,1}:=S_i\cap A_1$ for each $i\in[t]$.\\
Define an (unordered) partition $K:=K_1\cupdot\ldots\cupdot K_s$ with
$\CK:=\{K_1,\ldots,K_s\}$ such that:\\
$S_{i,1}=S_{j,1}$, if and only if
$(\Delta_i\rho_i,S_i),(\Delta_j\rho_j,S_j)\in K_\ell$ for
some $\ell\in[s]$.
\begin{cs}
\case{$s=t$}
\emph{(Thus $S_{i,1}\neq S_{j,1}$ for all $i\neq j$.)}\\
Recurse and return
$\Lambda:=\canSetHyper(K,A_1,\Delta\rho)$.

\case{$s=1$}
\emph{(Thus $S_{i,1}= S_{j,1}$ for all $i,j\in[t]$
and thus $S_i\cap A_2\neq S_j\cap A_2$ for all $i\neq j$.)}\\
Recurse and return
$\Lambda:=\canSetHyper(K,A_2,\Delta\rho)$.

\case{$1<s<t$}
\emph{(Thus $S_{i,1}= S_{j,1}$ for some $i\neq j$, but not
for all $i\neq j$.)}\\
\emph{(We compute a canonical labeling for $\CK$.)}\\
Compute $\Theta_i\tau_i:=\canSetHyper(K_i,A_2,\Delta\rho)$ for
each $i\in[s]$ recursively.\\
For each $K_i$ pick a $(\Delta_j\rho_j,S_j)\in K_i$ and set $R_i:=S_{j,1}$.\\
\emph{(Here~$R_i$ does not depend on the choice of $j$.)}\\
Let $\CK^\setHyper:=\{(\Theta_1\tau_1,R_1),\ldots,(\Theta_s\tau_s,R_s)\}$.\\
Define an \emph{ordered} partition
$\CK^\setHyper=\CK^\setHyper_1\cupdot\ldots\cupdot\CK^\setHyper_r$ 
such that\\
$(K_i,\Delta\rho)^{\tau_i}
\prec (K_j,\Delta\rho)^{\tau_j}$, if and only if
$(\Theta_i\tau_i,R_i)\in\CK^\setHyper_p$
and $(\Theta_j\tau_j,R_j)\in\CK^\setHyper_q$ for some~$p,q\in[r]$ with~$ p<q$.\\
Recurse and return
$\Lambda:=\canSetHyper((\CK^\setHyper_1,A_1),\ldots,(\CK^\setHyper_r,A_1);\Delta\rho)$.\\
\iterated
\end{cs}
\case{$\Delta$ is transitive on $A$}~\\
\transitive{\canSetHyper}{K}
\end{cs}

\begin{A}
Towards showing correctness of the algorithm, we first argue that Condition (A) remains satisfied in recursive calls.
In the intransitive case, Condition (A)
remains satisfied for
the recursive calls by construction of the partition of $K$.
In the transitive case, $K$ and $A$ remain unchanged,
and therefore Condition (A) also remains satisfied.
\end{A}

\begin{cl1}
As usual, Property (CL1) follows since 
all ordered sequences and
all partitions are defined
in an isomorphism-invariant way.
\end{cl1}

\begin{cl2}
Consider the Case $|A|\leq 1$.
If $|K|=1$, then $\aut((K,\Delta\rho))$
is equal to the intersection $\Delta_1\cap\stab(S_1)\cap\Delta$.
Analogously, if $|K|=2$, then $\aut((K,\Delta\rho))$
is equal to $\Delta_1\cap\stab(S_1)\cap\Delta_2\cap\stab(S_2)\cap\Delta$.

We consider the intransitive case.
In the Cases $s=t$ and $s=1$, (CL2) holds by induction.
We consider the Case $1<s<t$.
By induction, it holds that
$\Lambda$ defines
a canonical labeling for $(\CK^\setHyper_1,\ldots,\CK^\setHyper_r,\Delta\rho)$.
Since $R_i$ was chosen in an isomorphism-invariant way,
$\Lambda$ defines a canonical labeling for
$(\CJ^\setHyper_1,\ldots,\CJ^\setHyper_r,\Delta\rho)$
where $\CJ^\setHyper_i=\{\Lambda\mid (\Lambda,R)\in \CR^\setHyper_i\}$
for $i\in[r]$.
Because of the object replacement lemma (\cref{lem:rep}),
it holds that $\Lambda$ defines
a canonical labeling for $(\CJ_1,\ldots,\CJ_r,\Delta\rho)$
where $\CJ=\CJ_1\cupdot\ldots\cupdot\CJ_r$
such that
$K_i\in \CJ_\ell$, if and only if $\Theta_i\tau_i\in\CJ^\setHyper_\ell$.
Since $(\CJ_1,\ldots,\CJ_r)$ is an isomorphism-invariant \emph{ordered}
partition of $\CJ=\CK$, it holds that
$\Lambda$ defines a canonical labeling for $(\CK,\Delta\rho)$.
Again, $\CK=\{K_1,\ldots,K_s\}$ is an isomorphism-invariant
(unordered) partition
of $K=K_1\cupdot\ldots\cupdot K_s$
and therefore $\Lambda$ defines a canonical labeling for $(K,\Delta\rho)$.

The transitive case is similar to the analysis in \cref{lem:canMatch}.
\end{cl2}

\begin{runtime}
Let $A^*\subseteq A$ be a $\Delta$-orbit of maximal size.
We claim that the maximum number of recursive calls given these parameters $R(|A^*|,|A|,|K|)$ is at most~$T:=2^{6|A^*|}|A||K|^3$.

In the Case $|A|\leq 1$, we do not have further recursive calls.
The transitive case is similar to the analysis in \cref{lem:canMatch}.
In the intransitive cases if $s=t$ or $s=1$, there is only one recursive call and the size of $A$ decreases.

We consider the intransitive case when $1<s<t$.
Here, we have recursive calls on $K_1,\ldots,K_s$
and $\CK^\setHyper_1,\ldots\CK^\setHyper_r$.
The cardinalities of $|A^*|$ and $|A|$ never increase for recursive calls.
\begin{align*}
R(|A^*|,|A|,|K|)&\leq 1+
\sum_{i=1}^{r}R(|A_i^*|,|A_1|,|\CK^\setHyper_i|)+\sum_{i=1}^{s}
R(|A^*|,|A_2|,|K_i|)\\
&\overset{\mathclap{\text{induction}}}{\leq}
1+
\sum_{i=1}^{r}T(|A_i^*|,|A_1|,|\CK^\setHyper_i|)+\sum_{i=1}^{s}
T(|A^*|,|A_2|,|K_i|)\\
&\leq 1+
T(|A^*|,|A|,\sum_{i=1}^{r}|\CK^\setHyper_i|)+\sum_{i=1}^{s}
T(|A^*|,|A|,|K_i|)\\
&= 1+ 2^{6|A^*|}|A|( s^3 +\sum_{i=1}^{s} |K_i|^3)
\leq T,
\end{align*}
where for the last inequality we argue as follows: we use
that $s\neq t$ and thus
there is an $i\in[s]$ such that
$|K_i|\geq 2$.
Without loss of generality, we assume $|K_1|\geq 2$.
For $i\neq 1$, we use $|K_i|\geq 1$.
\begin{align*}
|K|^3=(\sum_{i=1}^{s} |K_i|)^3
&=
\sum_{i=1}^{s} |K_i|^3
+
\sum_{\substack{2\leq i,j,\ell \leq s\\|\{i,j,\ell\}|\geq 2}} |K_i||K_j||K_\ell|
+
\sum_{2\leq i,j\leq s} 3|K_1||K_i||K_j|
+
\sum_{2\leq i\leq s} 3|K_1|^2|K_i|\\
&\geq
\sum_{i=1}^{s} |K_i|^3+
((s-1)^3-(s-1))+6(s-1)^2+12(s-1)\\
&=\sum_{i=1}^{s} |K_i|^3+s^3+3s^2+2s-6
\overset{s>1}{>} \sum_{i=1}^{s} |K_i|^3+s^3
\end{align*}
We consider the running time for one single call without recursive costs.
Such a call can be computed
in time ${2^{\CO(k)}|V|^{\CO(1)}}$
where $k$ is the largest $\Delta$-orbit of $V$.
In total, we have a running time of at most
$T\cdot 2^{\CO(k)}|V|^{\CO(1)}
\subseteq 2^{\CO(k)}n^{\CO(1)}$
where
$n$ is the input size and
$k$ is the size of the largest $\Delta$-orbit of $V$.
\end{runtime}
\end{proof}

\begin{cor}\label{cor:hyperK}
Canonical labelings for hypergraphs
can be computed in time $2^{\CO(k)}n^{\CO(1)}$
where $n$ is the input size and
$k$ is the size of the largest color class of $V$.
\end{cor}

\begin{proof}
Given a
hypergraph with hyperedges $\{S_1,\ldots,S_t\}$
and a coloring $(C_1,\ldots,C_t)$
of $V=C_1\cupdot\ldots\cupdot C_t$,
we use the previous algorithm to compute
a canonical labeling for $(K,\Delta\rho)$
where
$K=\{(\lab(V),S_1),\ldots,(\lab(V),S_t)\}$
and $\Delta\rho=\{\lambda\in\lab(V)\mid\forall i,j\in[t],i<j \forall v_i\in
C_i,v_j\in C_j:
\lambda(v_i)<\lambda(v_j)\}$.
\end{proof}

%% file: figures/bundling.tex
\begin{tikzpicture}[scale=0.7,use Hobby shortcut]

\draw[xstep=6.0,ystep=8] (0,0) grid (12,8);
\draw (0.5,7.5) node {$A_1$};
\draw (6.5,7.5) node {$A_2$};


\newcommand{\halfedgeA}{
\begin{scope}
\clip (0,-0.1) rectangle (6,3.1);
\filldraw[closed] (6.5,3) .. (6,3) .. (5.5,3) .. (4,2) .. (3,2) .. (1,1) ..
(2,0) ..
(3,0.5) ..
(5.5,1) .. (6,1) ..
(6.5,1) .. (8,2);
\end{scope}}

\newcommand{\halfedgeB}{
\begin{scope}
\clip (0,0.3) rectangle (6,3.1);
\filldraw[closed] (6.5,3) .. (6,3) .. (5.5,3) .. (4.5,2) .. (3,2) .. (0.5,1.5)
..
(1.5,0.5) ..
(3,0.5) ..
(5.5,1) .. (6,1) ..
(6.5,1) .. (8,2);
\end{scope}}

\newcommand{\mirror}[1]{
\begin{scope}[yscale=-1,yshift=-4cm]
#1
\end{scope}}

\newcommand{\rightSide}[1]{
\begin{scope}[xscale=-1,xshift=-12cm]
#1
\end{scope}}

\begin{scope}[color=black,fill=\myGreen,fill opacity=0.3,yshift=4.1cm]
\halfedgeA
\halfedgeA
\rightSide{\mirror{\halfedgeB}}
\rightSide{\halfedgeB}
\end{scope}

\begin{scope}[color=black,fill=\myBlue,fill opacity=0.3,yshift=2.2cm]
\halfedgeB
\halfedgeB
\rightSide{\mirror{\halfedgeB}}
\rightSide{\halfedgeA}
\end{scope}

\begin{scope}[color=black,fill=\myRed,fill opacity=0.3,yshift=-0.2cm]
\mirror{\halfedgeA}
\mirror{\halfedgeA}
\rightSide{\mirror{\halfedgeA}}
\rightSide{\halfedgeB}
\end{scope}

\end{tikzpicture}

%% file: articles/multiCoset.tex
\section{Canonization of Sets}\label{sec:sets}

In this section, we are looking
for a canonical labeling function
for objects $\CX$
in the case where
$\CX=(\{\Delta_1\rho_1,\ldots,\Delta_t\rho_t\},\Delta\rho)$.
As we will see in \Cref{sec:obj}, canonical labelings
for objects in general can be reduced to this case
in polynomial time.

\problem{\canSet\label{prob:CL:Set}}
{$(J,\Delta\rho)\in\obj(V)$ where $J=\{\Delta_1\rho_1,\ldots,\Delta_t\rho_t\}$,
$\Delta_i\rho_i\leq\lab(V)$ for all $i\in[t]$,
$\Delta\rho\leq\lab(V)$ and $V$ is an unordered set}
{(J,\Delta\rho)}
{(J^\phi,\phi^{-1}\Delta\rho)}
{\{\delta\in\Delta\mid\exists\psi\in\sym(t)
\forall
i\in[t]:\delta^{-1}\Delta_i\rho_i
=\Delta_{\psi(i)}\rho_{\psi(i)}\}}

As usual, Conditions (CL1) and (CL2) coincide with the general
the conditions. In particular,
$\canSet(J,\Delta\rho)=\aut((J,\Delta\rho))\pi$ for some $\pi$.

\begin{lem}\label{lem:canSet}
A function $\canSet$
solving \cref{prob:CL:Set} can be computed in time
$2^{\CO(k)}n^{\CO(1)}$
where~$n$ is the input size and
$k$ is the size of the largest $\Delta$-orbit of $V$.
\end{lem}

\paragraph{Intuition for the Sets of Labeling Cosets Algorithm}
As it was the case for hypergraph algorithm, we need a
generalization of our problem for a recursive approach.
We extend the instances to sets of the form
$L=\{(\Delta_1\rho_1,\Theta_1\tau_1),\ldots,(\Delta_t\rho_t,\Theta_t\tau_t)\}$.
For the initial instance, $\Theta_i\tau_i$ is simply set to equal $\Delta_i\rho_i$.
The labeling cosets $\Theta_i\tau_i$ can be seen as an analogue to the
hyperedges $S_i$ in Problem~\ref{prob:CL:SetHyper} and are used to define a bundling of the instance $L$.
Compared to hyperedges $S_i$ the labeling cosets $\Theta_i\tau_i$ can describe 
global interdependencies.
While Condition (A) for hypergraphs describes a local distinctness,
for our new problem, we define adequate Conditions (C) and
(AC) describing a more refined distinctness of the cosets.
But again, the size of a set $A\subseteq V$ is used to measure our progress in
the sense that the size $|L|$ is bounded in terms of $|A|$.
Additionally, we will ensure a uniformity condition (Condition (D)) which
helps us in various situations.
In the algorithm, following our general strategy, we heavily use the partitioning techniques.

First, we consider the case in which at least one group $\Theta_i$ is
intransitive on the set $A$.
The uniformity (D) implies that all groups $\Theta_i$ must be
intransitive. In particular, each $\Theta_i$ can be associated with the
orbit $A_{i,1}\subseteq A$ that has minimal image under $\tau_i$.

If all orbits $A_{i,1}$ and $A_{j,1}$ are distinct for $i\neq j$, the sets
$A_{i,1}\subseteq A$ form a hypergraph and we apply our hypergraph algorithm.
If the orbits $A_{i,1}$ and $A_{j,1}$ are equal for all $i,j\in[t]$,
we can define a set $A_1:=A_{i,1}$ which does not depend on the choice of $i\in[t]$.
This gives us a canonical partition $A=A_1\cupdot A\setminus A_1$ of the set $A$ we are focusing on.
With a similar idea as for hypergraphs, we are able to use the partition of $A$
to define a partition $L=L_1\cupdot\ldots\cupdot L_s$
of $L$ into bundles.
In the extreme cases ($s=t$ or
$s=1$), we use the concrete definition of the partition to reduce the size of the set $A$, making progress.
The case of a proper partition can be handled by bundling and recursion as usual (one has to be
careful that (C) and (AC) are maintained for the recursive call).
Also in the remaining subcase where some of the orbits $A_{i,1}\subseteq A$
being equal to others but not to all, we can define a partition into bundles and recurse on that.

Second, we consider the case in which all groups $\Theta_i$ are transitive on the focus set $A$.
A crucial difference here is that instead
of splitting the underlying vertex $A$ by decomposing only one single group,
we split all the labeling cosets $\Theta_i\tau_i$ by decomposing them into left cosets.
This leads to the case of intransitive groups $\Theta_i$ which we already
handled.

\paragraph{Detailed Description of the Sets of Labeling Cosets Algorithm} 
Proving Lemma~\ref{lem:canSet}, we give a detailed description and analysis of the algorithm for sets of labeling cosets.

\begin{proof}[Proof of Lemma~\ref{lem:canSet}]
We generalize the problem and instead of $J$
we take as an input a set of pairs of labeling cosets
$L=\{(\Delta_1\rho_1,\Theta_1\tau_1),\ldots,(\Delta_t\rho_t,\Theta_t\tau_t)\}$.
For the purpose of recursion, we need some additional input parameters.
Specifically, we use subsets $A,C\subseteq V$ such that
$\Theta_i\leq\stab(A,C)$ for all $i\in[t]$.
Initially, we set $A=V$ and $C=\emptyset$
and $\Theta_i\tau_i=\Delta_i\rho_i$ for all $i\in[t]$. Furthermore, we require
that \begin{enumerate}[(AC)]
  \item[\textnormal{(C)}] $\Theta_i\tau_i|_C=
\Theta_j\tau_j|_C$ for all $i,j\in[t]$, and
  \item[\textnormal{(AC)}] $\Theta_i\tau_i|_{A\cup C}
\neq\Theta_j\tau_j|_{A\cup C}$
for all $i\neq j$.
\end{enumerate}

\algorithmDAN{\canSet(L,A,C,\Delta\rho)}
\begin{cs}

\case{$|L|=0$}
Return $\canInt(\lab(A)\uparrow^V,\lab(C)\uparrow^V;\Delta\rho)$.\\
\iterated
\case{$|A|\leq 1$}~\\
Rename the indices in $[t]$ such that
$A^{\tau_1}\prec\ldots\prec A^{\tau_t}$.\\
\emph{(The ordering is strict, for a proof see
(Correctness.) below the algorithm.)}\\
Return
$\Lambda:=\canInt(\Delta_1\rho_1,\ldots,\Delta_t\rho_t,\Theta_1\tau_1,\ldots,\Theta_t\tau_t;\Delta\rho)$.\\
\iterated
\end{cs}
Let $A_i^\ord:=A^{\tau_i},C_i^\ord:=C^{\tau_i}$ and
$\Theta_i^\ord:=(\Theta_i\tau_i)^{\tau_i}$ for each $i\in[t]$.
\begin{cs}
\case{$(A_i^\ord,C_i^\ord,\Theta_i^\ord)\neq(A_j^\ord,C_j^\ord,\Theta_j^\ord)$
for some $i,j\in[t]$}~\\
Define an \emph{ordered} partition $ L= L_1\cupdot\ldots\cupdot L_s$ such
that:\\
$(A_i^\ord,C_i^\ord,\Theta_i^\ord)\prec (A_j^\ord,C_j^\ord,\Theta_j^\ord)$,
if and only if\\
$(\Delta_i\rho_i,\Theta_i\tau_i)\in L_p$ and $(\Delta_j\rho_j,\Theta_j\tau_j)\in
 L_q$ for some $p,q\in[s]$ with~$p<q$.\\
Recurse and return $\Lambda:=\canSet((L_1,A,C),\ldots,
(L_s,A,C);\Delta\rho)$.\\
\iterated
\end{cs}
\emph{Now, we have the following property (D):
$(A_i^\ord,C_i^\ord,\Theta_i^\ord)=(A_j^\ord,C_j^\ord,\Theta_j^\ord)$ for all
$i,j\in[t]$.}
\begin{cs}
\case{$\Theta_i$ is intransitive on $A$ for some (and because of (D) for
all) $i\in[t]$}~\\
For each $i\in[t]$ write $A=A_{i,1}\cupdot A_{i,2}$ where
$A_{i,1}$ is a $\Theta_i$-orbit and such that
$A_{i,1}^{\tau_i}$ is minimal.\\
\minimal
\begin{cs}
\case{$A_{i,1}= A_{j,1}$ for all $i,j\in[t]$}~\\
Let $A_1:=A_{i,1}$ for some (and thus all) $i\in[t]$.\\
Let $\Lambda_i:=\Theta_i\tau_i|_{A_1\cup C}$ for each $i\in[t]$.\\
Define an (unordered) partition $L= L_1\cupdot\ldots\cupdot L_s$ with
$\CL:=\{L_1,\ldots,L_s\}$ such
that:\\
$\Lambda_i=\Lambda_j$,
if and only if
$(\Delta_i\rho_i,\Theta_i\tau_i),(\Delta_j\rho_j,\Theta_j\tau_j)\in
L_\ell$ for some $\ell\in[s]$.

\begin{cs}
\case{$s=t$}
\emph{(Thus $\Lambda_i\neq\Lambda_j$ for all $i\neq j$.)}\\
Recurse and return $\Lambda:=\canSet(L,A_1,C,\Delta\rho)$

\case{$s=1$}
\emph{(Thus $\Lambda_i=\Lambda_j$ for all $i,j\in[t]$.)}\\
Recurse and return $\Lambda:=\canSet(L,A_2,A_1\cup C,\Delta\rho)$

\case{$1<s<t$}
\emph{(Thus $\Lambda_i=\Lambda_j$ for some $i\neq j$, but not
for all $i\neq j$.)}\\
\emph{(We compute a canonical labeling for $\CL$.)}\\
Compute $\Pi_i\eta_i:=\canSet(L_i,A_2,A_1\cup C,\Delta\rho)$ for
each $i\in[s]$ recursively.\\
For each $L_i$ pick a $(\Delta_j\rho_j,\Theta_j\tau_j)\in L_i$ and set
$\hat\Gamma_i:=\Lambda_j$.\\
\emph{(The set~$\hat\Gamma_i$ does not depend on the choice of $j$, but $\hat\Gamma_i$ is not a
labeling coset.)}\\
Let
$\Gamma_i:=\{\gamma\in\lab(V)\mid \gamma|_{A_1\cup C}\in\hat\Gamma_i\}\leq\lab(V)$.\\
\emph{(The definition of $\Gamma_i$ rectifies that $\hat\Gamma_i$ is not a labeling coset.)}\
\\
Set $\CL^\set:=\{(\Pi_1\eta_1,\Gamma_1),\ldots,(\Pi_s\eta_s,\Gamma_s)\}$.\\
Define an \emph{ordered} partition
$\CL^\set=\CL_1^\set\cupdot\ldots\cupdot\CL^\set_r$ such that:\\
$(L_i,\Delta\rho)^{\eta_i}
\prec (L_j,\Delta\rho)^{\eta_j}$, if and only if
$(\Pi_i\eta_i,\Gamma_i)\in\CL^\set_p$ and $(\Pi_j\eta_j,\Gamma_j)\in\CL^\set_q$
for some~$p,q\in[r]$ with~$p<q$.\\
Recurse and return
$\Lambda:=\canSet((\CL^\set_1,A_1, C),\ldots,(\CL^\set_r,A_1,C);\Delta\rho)$.\\
\iterated
\end{cs}

\case{$A_{i,1}\neq A_{j,1}$ for some $i,j\in[t]$}~\\
Define an (unordered) partition $L= L_1\cupdot\ldots\cupdot L_s$ with
$\CL:=\{L_1,\ldots,L_s\}$ such
that:\\
$A_{i,1}=A_{j,1}$,
if and only if
$(\Delta_i\rho_i,\Theta_i\tau_i),(\Delta_j\rho_j,\Theta_j\tau_j)\in
L_\ell$ for some $\ell\in[s]$.\\
\emph{(Observe that $1<s\leq t$.)}\\
\emph{(We compute a canonical labeling for $\CL$.)}\\
Compute $\Pi_i\eta_i:=\canSet(L_i,A,C,\Delta\rho)$ for
each $i\in[s]$ recursively.\\
For each $L_i$ pick a $(\Delta_j\rho_j,\Theta_j\tau_j)\in L_i$ and set
$S_i:=A_{j,1}$.\\
\emph{(The set~$S_i$ does not depend on the choice of $j$.)}\\
Set $\CL^\setHyper:=\{(\Pi_1\eta_1,S_1),\ldots,(\Pi_s\eta_s,S_s)\}$.\\
Define an \emph{ordered} partition
$\CL^\setHyper=\CL_1^\setHyper\cupdot\ldots\cupdot\CL^\setHyper_r$ such that:\\
$(L_i,\Delta\rho)^{\eta_i}
\prec (L_j,\Delta\rho)^{\eta_j}$, if and only if
$(\Pi_i\eta_i,S_i)\in\CL^\setHyper_p$ and $(\Pi_j\eta_j,S_j)\in\CL^\setHyper_q$
and $p<q\in[r]$.\\
Return
$\Lambda:=\canSetHyper((\CL^\setHyper_1,V).\ldots,(\CL^\setHyper_r,V);\Delta\rho)$
using \cref{lem:canSetHyper}.\\
\iterated
\end{cs}

\case{$\Theta_i$ is transitive on $A$ for some
(and because of (D) for all)
$i\in[t]$}~\\
Let $A^\ord:=A_i^\ord,C^\ord:=C_i^\ord$ and
$\Theta^\ord:=\Theta_i^\ord$ for some $i\in[t]$.\\
\emph{(Because of (D), the definitions do not depend on the choice of $i\in[t]$.)}\\
Partition $A^\ord=A^\ord_1\cupdot A^\ord_2$ such that
$|A^\ord_1|=\lfloor|A^\ord|/2\rfloor$ and
$A^\ord_1$ is minimal.\\
\minimal\\
Define $\Psi^\ord:=\stab_{\Theta^\ord}(A^\ord_1,A^\ord_2)$.\\
\emph{(The group can be computed using a membership test as
stated in the preliminaries.)}\\
Decompose $\Theta^\ord$ into left cosets of $\Psi^\ord$ and write
$\Theta^\ord=\bigcup_{i\in[s]}\theta^\ord_i\Psi^\ord$.\\
Let $\Theta_C^\ord:=\Theta^\ord|_{C^\ord}$ and
$\Psi_C^\ord:=\Psi^\ord|_{C^\ord}$.\\
Decompose $\Theta_C^\ord$ into
left cosets of $\Psi_C^\ord$ and write
$\Theta_C^\ord=\bigcup_{i\in[r]}\theta_{C,i}^\ord\Psi_C^\ord$.\\
\emph{(Observe that $r\leq s$.)}\\
For some $j\in[t]$, define a set of cosets $N:=\{\Gamma_1,\ldots,\Gamma_r\}$
where $\Gamma_i:=\tau_j|_C\theta_{C,i}^\ord\Psi_C^\ord$.\\
\emph{(Because of (C), the set $N$ does not depend on the choice of
$j\in[t]$.)}\\
Define
$L_i:=\{(\Delta_j\rho_j,\tau_j\theta^\ord_\ell\Psi^\ord)\mid
j\in[t],\ell\in[s],(\tau_j\theta^\ord_\ell\Psi^\ord)|_C=\Gamma_i\}$ for each
$i\in[r]$.\\
Compute $\Pi_i\eta_i:=\canSet(L_i,A,C,\Delta\rho)$ for each
$i\in [r]$ recursively.\\
Rename the indices of elements in $N=\{\Gamma_1,\ldots,\Gamma_r\}$ such that:\\
$(L,\Delta\rho)^{\eta_1}=\ldots=(L,\Delta\rho)^{\eta_q}
\prec(L,\Delta\rho)^{\eta_{q+1}}\preceq\ldots\preceq
(L,\Delta\rho)^{\eta_r}$.\\
\emph{(The ordering does not depend on the choice of the
$\eta_i$, for a proof see (CL1.).)}\\
Return $\Lambda:=
\langle\Pi_1\eta_1,\ldots,\Pi_q\eta_q\rangle$.\\
\emph{(This is the smallest coset containing these cosets as defined
in the preliminaries.)}
\end{cs}

\begin{correctness}
We consider  the case $|A|\leq 1$.
We show that we obtain a strict ordering indeed.
If $A=\emptyset$, then because of (C) and (AC),
it holds that $|L|=1$, so the ordering is strict.
Let us assume that $|A|=1$.
Since $A$ is $\Theta_i$-invariant and consists of one element,
it holds that $\Theta_i|_{A\cup C}$ is a direct product
of $\Theta_i|_A$ and $\Theta_i|_C$ for all $i\in[t]$.
But because of (C), it holds that $(\Theta_i\tau_i)|_A\neq
(\Theta_j\tau_j)|_A$ for all $i\neq j$
and therefore the ordering is strict.
\end{correctness}

\begin{AC}
We show that (C) and (AC) remain satisfied for the recursive calls.
Consider the case $(A_i^\ord,C_i^\ord,\Theta_i^\ord)\neq(A_j^\ord,C_j^\ord,\Theta_j^\ord)$.
Observe that if (C) and (AC) hold for the set $L$,
then they also hold for each subset of $L_i\subseteq L$.
Therefore, in this case (C) and (AC) remain satisfied.

Consider the intransitive case, when $A_{i,1}=A_{j,1}$.
In the cases $s=t$ and $s=1$, the conditions (C) and (AC)
 hold for chosen subsets of~$A$ by construction.
In the Case $1<s<t$, the labeling cosets $\Gamma_i$
satisfy (C) and (AC) by construction of the partition.

Consider the intransitive case, when $A_{i,1}\neq A_{j,1}$.
Here, we only have recursive calls for subsets $L_i\subseteq L$ with~$A$ and~$C$ unchanged.

Consider the transitive case.
Condition (C) holds by construction.
Next, we show that also (AC) remains satisfied.
Condition (D) implies
$\Theta^\ord_i|_{A^\ord_i\cup C^\ord_i}=\Theta^\ord_j|_{A^\ord_j\cup C^\ord_j}$
which means that $\Theta_i\tau_i|_{A\cup
C}$ and $\Theta_j\tau_j|_{A\cup C}$ are cosets of the same group.
Together with (AC), this implies that $\Theta_i\tau_i|_{A\cup
C}\cap\Theta_j\tau_j|_{A\cup C}=\emptyset$.
For this reason, also
subsets of these cosets are disjoint, i.e.,
$\tau_i\theta^\ord_\ell\Psi^\ord|_{A\cup C}
\cap\tau_{j}\theta^\ord_{\ell'}\Psi^\ord
|_{A\cup C}=\emptyset$ for
$i\neq j$ and $\ell,\ell'\in[s]$.
The same disjointness holds
for $i=j$ and $\ell\neq\ell'$
since $(\theta^\ord_\ell|_{A^\ord\cup C^\ord})^{-1}
\theta^\ord_{\ell'}|_{A^\ord\cup C^\ord}\in\Psi^\ord|_{A^\ord\cup C^\ord}$
implies $\ell=\ell'$.
In particular, the cosets
$\tau_i\theta^\ord_\ell\Psi^\ord|_{A\cup C}
$ and $\tau_{j}\theta^\ord_{\ell'}\Psi^\ord
|_{A\cup C}$
are distinct
for $(i,\ell)\neq(j,\ell')$.
\end{AC}

\begin{cl1}
This is not hard to see as
almost all ordered sequences and
all partitions are defined
in an isomorphism-invariant way.
There is only one case that needs attention.
In the transitive case, we define an ordering
that might a priori depend on the choice of $\eta_i$ for
$i\in[r]$.
To prove that this is indeed not the case,
we claim that $\aut(L_i)\subseteq
\aut(L)$.
Let $\sigma\in\aut(L_i)$ for $i\in[r]$.
Let $j\in[t]$.
There is at least one $\ell\in[s]$ such that
$(\Delta_j\rho_j,\tau_j\theta^\ord_{\ell}\Psi^\ord)\in L_i$
since
$\{(\tau_j\theta^\ord_\ell\Psi^\ord)|_C\mid \ell\in[s]\}=
\{\tau_j|_C\theta_{C,i}^\ord\Psi_C^\ord\mid i\in[r]\}=
\{\Gamma_1,\ldots,\Gamma_r\}$.
By definition of $\aut(L_i)$, there are $j'\in[t]$
and $\ell'\in[s]$ such that
$(\Delta_j\rho_j,\tau_j\theta^\ord_{\ell}\Psi^\ord)^\sigma=
(\Delta_{j'}\rho_{j'},\tau_{j'}\theta^\ord_{\ell'}\Psi^\ord)$.
This implies
$(\Delta_j\rho_j,\Theta_j\tau_j)^\sigma=
(\Delta_{j'}\rho_{j'},\Theta_{j'}\tau_{j'})$.
Since for all $j\in[t]$, there is such a $j'\in[t]$,
this implies $\sigma\in\aut(L)$.

\end{cl1}

\begin{cl2}
In the Case $|A|\leq 1$, the strict ordering of the
sequence implies
that $\sigma^{-1}\Theta_i\tau_i\neq \Theta_j\tau_j$
for all $\sigma\in\aut(A)$ and all $i\neq j$.
Therefore, $\aut((L,A,C,\Delta\rho))
=\Delta_1\cap\ldots\cap\Delta_t\cap\Theta_1\cap\ldots\cap\Theta_t\cap\stab(A,C)\cap\Delta$
where $\stab(A,C)$ can be omitted since $\Theta_1\leq\stab(A,C)$.

We consider the Case $(A_i^\ord,C_i^\ord,\Theta_i^\ord)\neq(A_j^\ord,C_j^\ord,\Theta_j^\ord)$.
Here, $\CL=(L_1,\ldots,L_s)$ is an isomorphism-invariant \emph{ordered}
partition of $L=L_1\cupdot\ldots\cupdot L_s$ and therefore $\aut(\CL)=\aut(L)$.

Consider the intransitive case, when $A_{i,1}=A_{j,1}$
and $s=t$ or $s=1$.
Condition (D) and $L\neq\emptyset$ also ensure $\aut(L)\leq\stab(A,C)$.
Therefore, changing these parameters
in an isomorphism-invariant way
does not affect the automorphism group.
In the Case $1<s<t$, the labeling coset~$\Lambda$ defines a canonical labeling
for $(\CL_1^\set,\ldots,\CL_r^\set,\Delta\rho)$ by induction.
Since $\Gamma_i$ was chosen in an isomorphism-invariant way,
$\Lambda$ defines a canonical labeling for
$(\CJ^\set_1,\ldots,\CJ^\set_r,\Delta\rho)$
where $\CJ^\set_i=\{\Lambda\mid (\Lambda,\Gamma)\in \CL^\set_i\}$
for $i\in[r]$.
Because of the object replacement lemma (\cref{lem:rep}),
it holds that $\Lambda$ defines
a canonical labeling for $(\CJ_1,\ldots,\CJ_r,\Delta\rho)$
where $\CJ=\CJ_1\cupdot\ldots\cupdot\CJ_r$
such that
$L_i\in \CJ_\ell$, if and only if $\Pi_i\eta_i\in\CJ^\set_\ell$.
Since $(\CJ_1,\ldots,\CJ_r)$ is an isomorphism-invariant \emph{ordered} partition
of $\CJ=\CL$, it holds that
$\Lambda$ defines a canonical labeling for $(\CL,\Delta\rho)$.
Again, $\CL=\{L_1,\ldots,L_s\}$ is an isomorphism-invariant
(unordered) partition
of $L=L_1\cupdot\ldots\cupdot L_s$
and therefore $\Lambda$ defines a canonical labeling for $(L,\Delta\rho)$.

Consider the intransitive case, when $A_{i,1}\neq A_{j,1}$.
Here, the proof is analogous to the previous Case $1<s<t$.
Also here, $\Lambda$ defines a canonical labeling
of $(\CL_1^\setHyper,\ldots,\CL_r^\setHyper,\Delta\rho)$.
Since $S_i$ was chosen in an isomorphism-invariant way,
the labeling coset~$\Lambda$ defines a canonical labeling for
$(\CJ^\setHyper_1,\ldots,\CJ^\setHyper_r,\Delta\rho)$
where $\CJ^\setHyper_i=\{\Lambda\mid (\Lambda,S)\in \CL^\setHyper_i\}$
for $i\in[r]$.
Because of the object replacement lemma (\cref{lem:rep}),
it holds that $\Lambda$ defines
a canonical labeling for $(\CJ_1,\ldots,\CJ_r,\Delta\rho)$
where $\CJ=\CJ_1\cupdot\ldots\cupdot\CJ_r$
such that
$L_i\in \CJ_\ell$, if and only if $\Pi_i\eta_i\in\CJ^\setHyper_\ell$.
Since $(\CJ_1,\ldots,\CJ_r)$ is an isomorphism-invariant \emph{ordered} partition
of $\CJ=\CL$, it holds that
$\Lambda$ defines a canonical labeling for $(\CL,\Delta\rho)$.
Again, $\CL=\{L_1,\ldots,L_s\}$ is an isomorphism-invariant
(unordered) partition
of $L=L_1\cupdot\ldots\cupdot L_s$
and therefore $\Lambda$ defines a canonical labeling for $(L,\Delta\rho)$.

Consider the transitive case.
We claim that $\Lambda=\aut((L,\Delta\rho))\pi$ for some
(and thus for all) $\pi\in\Lambda$.
The inclusion
$\aut((L,\Delta\rho))\pi\subseteq\Lambda$ follows from
(CL1) of this algorithm.
It remains to show $\Lambda\subseteq\aut((L,\Delta\rho))\pi$.
We need to show that  $\eta_i\eta_j^{-1}$
is an element in $\aut((L,\Delta\rho))$ for all $i,j\in[q]$.
The membership
follows from the fact that $(L,\Delta\rho)^{\eta_i}=(L,\Delta\rho)^{\eta_j}$. 
\end{cl2}

\begin{runtime}
Let $A^*\subseteq A$ be a $\Theta_i$-orbit of maximal size
over all $i\in[t]$.
We claim that the maximum number of recursive calls given these parameters $R(|A^*|,|A|,|L|)$  is at most $T:=2^{14|A^*|}|A||L|^3$.

In the
Case $(A_i^\ord,C_i^\ord,\Theta_i^\ord)\neq(A_j^\ord,C_j^\ord,\Theta_j^\ord)$
and in the
intransitive cases, we make progress on $|A|$ or $|L|$
similar to the analysis of \cref{lem:canSetHyper}.

We consider the transitive case.
Here, we have recursive calls on $L_1,\ldots,L_r$.
Observe that $\sum^r_{i=1} |L_i|=s|L| \leq 2^{|A|}|L|$.
For the recursive calls, we make progress on $|A^*|=|A|$.
\begin{align*}
R(|A|,|A|,|L|)&\leq 1+ \sum^r_{i=1}R(\lceil|A|/2\rceil,|A|,|L_i|)\\
&\overset{\mathclap{\text{induction}}}{\leq}
1+ \sum^r_{i=1}T(\lceil|A|/2\rceil,|A|,|L_i|)\\
&\leq
1+ T(|A|/2+1/2,|A|,\sum^r_{i=1}|L_i|)
\leq1+2^{10|A|+7}|A||L|^3
\overset{2\leq |A|}{\leq}
T.
\end{align*}
This gives at most $T$ recursive
calls for the instance $(L,A,C,\Delta\rho)$.

Summing up, our argument so far
shows a bound of $2^{14|A^*|}|A||L|^3$ on the number of recursive calls, where~$A^*$ is the maximum size $\Theta_i$-orbit within~$A$. Applying the algorithm on an instance $(J,\Delta\rho)$  gives us a bound of $2^{\CO(k)}|V||J|^3$ on the number of recursive calls, where $k$ is the size
of the largest $\Delta_i$-orbit of $V$
over all $i\in[t]$.
However, the running time claimed by the theorem is in terms of the largest $\Delta$-orbit rather than the $\Delta_i$-orbits. 
To achieve this, we add a preprocessing step before the algorithm that ensures that the~$\Delta_i$-orbits are no larger than the~$\Delta$-orbits.

Given the instance $(J,\Delta\rho)$, this can be done as follows.
Compute $\Delta_i'\rho_i':=\canInt(\Delta_i\rho_i,\Delta\rho)$
for each $\Delta_i\rho_i\in J$.
Set $J':=\{\Delta_1'\rho_1',\ldots,\Delta_t'\rho_t'\}$.
Then define an \emph{ordered} partition 
$J'=J'_1\cup\ldots\cup J'_s$ such that:
$(\Delta_i\rho_i,\Delta\rho)^{\rho_i'}
\prec (\Delta_j\rho_j,\Delta\rho)^{\rho_j'}$, if and only if
$\Delta_i'\rho_i'\in J'_p$ and $\Delta_j'\rho_j'\in J'_q$ and
$p<q\in[s]$.\\
This gives a new (iterated) instance $(J'_1,\ldots,J'_s,\Delta\rho)$
which can be processed by our algorithm
that we just presented.
A canonical labeling for
the new instance defines
a canonical labeling also for $(J,\Delta\rho)$
by the object replacement lemma.
Furthermore,
the new instance has the property that
for all $\Delta_i'\rho_i'\in J'_j$ and all $j\in[s]$ the~$\Delta'_i$-orbits are not larger than the~$\Delta$-orbits.
\end{runtime}
\end{proof}

%% file: articles/hereditarilyFiniteObjects.tex
\section{Canonization of Objects}\label{sec:obj}

As argued in \cref{lem:rep}, for the purpose of a canonical
labeling, objects can be replaced with their canonical labelings.
With \cref{lem:canSet}, we are able to compute canonical labelings
for a set of labeling cosets.
In this section, we will combine both techniques to compute
a canonical labeling function for objects in general.

\problem{\canObj\label{prob:CL:Obj}}
{$(\CX,\Delta\rho)\in\obj(V)$ where
$\Delta\rho\leq\lab(V)$ and $V$ is an unordered set}
{(\CX,\Delta\rho)}
{(\CX^\phi,\phi^{-1}\Delta\rho)}
{\aut((\CX,\Delta\rho))}

\begin{theo}\label{theo:obj}

A function $\canObj$ solving
\cref{prob:CL:Obj}
can be computed in time
$2^{\CO(k)}n^{\CO(1)}$
where $n$ is the input size and
$k$ is the size of the largest $\Delta$-orbit of $V$.
\end{theo}

\begin{proof}

\algorithmDAN{\canObj(\CX,\Delta\rho)}
\begin{cs}

\case{$\CX=v\in V$} Return $\canPoint(v,\Delta\rho)$.

\case{$\CX=\Theta\tau\leq\lab(V)$} Return $\canInt(\Theta\tau,\Delta\rho)$.

\case{$\CX=(X_1,\ldots,X_t)$}~\\
Compute $\Delta_i\rho_i:=\canObj(X_i,\Delta\rho)$
for each $i\in[t]$ recursively.\\
Return
$\Lambda:=\canInt(\Delta_1\rho_1,\ldots,\Delta_t\rho_t;\Delta\rho)$
using the algorithm from \cref{lem:canInt}.\\
\iterated

\case{$\CX=\{X_1,\ldots,X_t\}$}~\\
Compute $\Delta_i\rho_i:=\canObj(X_i,\Delta\rho)$
for each $i\in[t]$ recursively.\\
Set $\CX^\set:=\{\Delta_1\rho_1,\ldots,\Delta_t\rho_t\}$.\\
Define an \emph{ordered} partition
$\CX^\set=\CX^\set_1\cup\ldots\cup\CX^\set_s$
such that:\\
$X_i^{\rho_i}
\prec X_j^{\rho_j}$,
if and only if
$\Delta_i\rho_i\in\CX^\set_p$ and $\Delta_j\rho_j\in\CX^\set_q$ for some
$p,q\in[s]$ with~$p<q$.\\
Return $\Lambda:=\canSet(\CX^\set_1,\ldots,\CX^\set_s;\Delta\rho)$
using the algorithm from \cref{lem:canSet}.\\
\iterated
\end{cs}

\begin{cl1}
This is easy to see as the partition of $\CX^\set$
is defined in an isomorphism-invariant way.
\end{cl1}

\begin{cl2}
The Case $\CX=\{X_1,\ldots,X_t\}$ follows
from the object replacement lemma (\cref{lem:rep}). The other cases use
canonical labeling functions we described in previous sections.
\end{cl2}

\begin{runtime}
With a dynamic programming approach, we build up a table
in which we store a
canonical labeling for each $(\CY,\Delta\rho)$
with $\CY\in\tranCl(\CX)$. Note that in each recursive call to~$\canObj$
the coset~$\Delta \rho$ remains unchanged and thus each recursive call is indeed
of the form $(\CY,\Delta\rho)$ with $\CY\in\tranCl(\CX)$. We thus get at
most~$|\tranCl(\CX)|\in \CO(n)$ recursive calls.
For the (non-recursive) calls to~$\canInt$ and~$\canSet$, we use the respective
algorithms from the previous sections to compute the solutions within the desired running time.
\end{runtime}
\end{proof}

\begin{cor}
Canonical labelings for combinatorial
objects
can be computed in time $2^{\CO(k)}n^{\CO(1)}$
where $n$ is the input size and
$k$ is the size of the largest color class of $V$.
\end{cor}

\begin{proof}
Given an
object $\CX\in\obj(V)$
and a coloring $(C_1,\ldots,C_t)$
of $V=C_1\cupdot\ldots\cupdot C_t$,
we use the previous algorithm to compute
a canonical labeling for $(\CX,\Delta\rho)$
where $\Delta\rho=\{\lambda\in\lab(V)\mid\forall i,j\in[t],i<j \forall v_i\in
C_i,v_j\in C_j:
\lambda(v_i)<\lambda(v_j)\}$.
\end{proof}

\paragraph{Canonization of Strings and Codes}
A \emph{string} or \emph{code word} with positions $V$
over a finite alphabet $\Sigma$ is a function
$\mathfrak{x}:V\to\Sigma$.
A \emph{code} is a set of code words.

Isomorphism of codes, also known as code equivalence,
was considered in \cite{DBLP:conf/soda/BabaiCGQ11}
and in \cite{DBLP:conf/icalp/BabaiCQ12}.
With our general result for objects, we
achieve the same time bound even for canonization.

\begin{cor}
Canonical labelings for codes
can be computed in time $2^{\CO(|V|)}|\CA|^{\CO(1)}$
where $|V|$ is the length of the code words and
$|\CA|$ is the size of the code (i.e., the number of code words).
\end{cor}

\begin{proof}
We need to explain how to encode a string with our formalism.
Observe that $\Sigma$ is an alphabet which has to be fixed
pointwise by
automorphisms.
For this reason, we can encode the elements in $\Sigma$
as $\emptyset,\{\emptyset\},\{\{\emptyset\}\},\ldots$ and so on.
A string $\mathfrak{x}$ can be encoded as
a function and thus as a set of pairs $\{(v,\mathfrak{x}(v))\mid v\in V\}$.
\end{proof}

We can also draw conclusions for canonization of permutation groups up to
permutational isomorphism.

\begin{cor}
Canonical labelings for permutation groups (up to permutational isomorphism)
can be computed in time $2^{\CO(|V|)}|\Delta|^{\CO(1)}$
where $V$ is the permutation domain and
$|\Delta|$ is the order of the group.
\end{cor}
\begin{proof}
We encode an element~$\delta\in \Delta\leq \sym(V)$ explicitly as permutation
over~$V$ by using the set~$M(\delta) := \{(v,v^\delta)\mid v\in
V\}$. Canonizing~$\{M(\delta)\mid \delta\in \Delta\}$ is the same as
canonizing the group $\Delta$ up to permutational isomorphism.
\end{proof}

%% file: articles/canGenSets.tex
\section{Canonical Generating Sets}\label{sec:Group}

The algorithms we have described throughout the paper canonize combinatorial
objects.
When given an unordered object as an input, they produce a canonical ordered
object as an output. However, this does not immediately give us a canonical
encoding of the output (i.e., is does not provide a canonical output string) as there may be multiple ways to represent the same ordered object.
Indeed, in \cref{sec:comb:objs:and:lab:cos},
we described how we represent objects
using generating sets and colored directed graphs.
Usually it is not very crucial which encoding we
use to translate the generating sets
and colored directed graphs into a binary string
which can be processed by an algorithm or a Turing machine.
However, especially in the context of canonization, it would be desirable to ensure that
if two implicitly given ordered objects $\CX,\CY\in\obj(\{1,\ldots,|V|\})$
are equal, then they also have the same string encoding~$\enc(\CX)=\enc(\CY)$.
For explicitly given objects $\CX$ over $\{1,\ldots,|V|\}$,
this can be achieved since
all elements in $\tranCl(\CX)$
are linearly ordered by ``$\prec$''.
Therefore, the crucial question is how to uniquely encode the
implicitly given labeling cosets $\Delta\rho\leq\lab(\{1,\ldots,|V|\})$.
To answer this question, we make use of 
canonical generating sets.

\begin{lem}[\cite{allender2018minimum}, Lemma 6.2, \cite{tree-width}, Lemma
21]\label{lem:canGen} There is a polynomial-time algorithm that, given a labeling coset $\Delta\rho\leq\lab(\{1,..,|V|\})$
via a generating set,
computes a generating set for $\Delta\rho$.
The output only depends
on $\Delta\rho$ (and not on the given generating set).
\end{lem}

With this lemma,
we can find a canonical generating set
for labeling cosets over natural numbers.
Using the order ``$\prec$'', we can thus compute for every ordered object a
string that uniquely encodes the object.
\begin{lem}\label{lem:enc}
There is an injective encoding~$\enc : \obj(\{1,\ldots
|V|\}) \to \{0,1\}^*$ that maps ordered objects to (0-1)-strings.
The encoding~$\enc$ and its inverse are
polynomial-time computable (for objects given
via generating sets and colored directed graphs).
\end{lem}

There is a second application of canonical generating sets in our context. For
this, let us consider the bottleneck of our canonization algorithms. The recursions
perform efficiently, except for
the transitive cases.
The concrete situation is that
we have some group $\Delta^\ord\leq\sym(\{1,\ldots,|V|\})$
acting transitively on
a set $A^\ord\subseteq\{1,\ldots,|V|\}$.
In this case, the algorithm computes
a subgroup $\Psi^\ord\leq\Delta^\ord$,
decomposes $\Delta^\ord$ into cosets of $\Psi^\ord$
and recurses.
Here, the branching degree of the algorithm
is the index of the subgroup $\Psi^\ord$ in $\Delta^\ord$.
However, progress is made since the
$\Psi^\ord$-orbits on $A^\ord$
are smaller than the $\Delta^\ord$-orbits
on $A^\ord$.
Luks's group theoretic framework~\cite{DBLP:journals/jcss/Luks82}, which solves isomorphism for bounded degree graphs in polynomial time,
attacks this bottleneck.
He used group theoretic insights
to argue that if $\Delta^\ord$ is of a group
of a certain type (a so called $\Gamma_d$-group), then there is a subgroup
of relatively small index and small orbit size.
Babai's algorithm in turn, attacks the
bottleneck of Luks's approach which is
characterized by \emph{giant} homomorphisms
from $\Delta^\ord$ to some symmetric group.
In the rest of this section, we discuss how these
approaches can be employed within our canonization framework
by the use of canonical generating sets.

A block system $\CB=\{B_1,\ldots,B_t\}$ for a
group $\Delta\leq\sym(V)$
is a $\Delta$-invariant partition $\CB$ of $V=B_1\cupdot\ldots\cupdot B_t$.
For a group $\Delta\leq\sym(V)$, the composition width
of $\Delta$, denoted as $\cw(\Delta)$, is
the smallest integer $k$ such that
all composition
factors of $\Delta$ are isomorphic to a subgroup of $\sym(k)$.
Luks's approach exploits that
one can efficiently compute a minimal block system for a given group.
One then computes the subgroup $\Psi$ that stabilizes all the blocks setwise.
Babai, Cameron and P\'alfy~\cite{MR679977} showed
that
for groups of bounded composition width
such a $\Psi$
is of relatively small index.
Therefore, going down to the subgroup $\Psi$ is not too costly,
but still reduces the size of the orbits in the recursive calls significantly.
Together these results prove the following.

\begin{theo}[\cite{DBLP:journals/jcss/Luks82} and
\cite{MR679977},\cite{DBLP:conf/focs/BabaiKL83}]\label{lem:luks} Let
$\Delta\leq\sym(V)$ be a group
that is transitive on a set $A\subseteq V$.
There is a subgroup $\Psi\leq \Delta$
and $b\in\NN$ such that:
\begin{enumerate}
  \item[(1)] The size of the $\Psi$-orbits on $A$ is bounded by $b$.  
  \item[(2)] The index $[\Delta:\Psi]$ is bounded by
  $(|A|/b)^{\CO(\cw(\Delta))}$.
\end{enumerate}
Furthermore, there is an
algorithm $\LA$ that given $A\subseteq V$ and a generating set for $\Delta$,
computes a generating set for a subgroup $\Psi\leq\Delta$.
The algorithm runs in time $|V|^{\CO(\cw(\Delta))}$.
\end{theo}

\begin{proof}[Proof sketch]
Let $\CB=\{B_1,\ldots,B_t\}$ be a minimal block system for $A\subseteq V$
and let $b:=|A|/t=|B_1|$.
Define $\Psi:=\stab_\Delta(B_1,\ldots,B_t)$.
The bound on the index follows from a bound of the size of
primitive permutation groups of degree $t$.
While Babai, Cameron, P\'alfy \cite{MR679977}
implicitly proved a bound of $t^{\omega(\cw(\Delta))}$
where $\omega\in\CO(n\log n)$, it
was later observed that $\omega$ can be chosen linear
as stated in
\cite{DBLP:conf/focs/BabaiKL83}, see \cite{liebeck1999simple}.
\end{proof}

For us, a crucial detail here is that the minimal block system and thus $\Psi$ is not unique.
For an isomorphism algorithm,
it is usually sufficient to
compute an arbitrary subgroup $\Psi\leq \Delta$
with Properties (1) and (2).
However, for the purpose of canonization,
we will need
that for each group $\Delta$
one particular subgroup $\Psi\leq \Delta$ will be computed
and the choice of $\Psi$ has to be consistent across different
calls of algorithm $\LA$.
More precisely, we need that for two
inputs $I$ and $I'$
which are both encodings for the pair $(A,\Delta)$,
the algorithm $\LA$ with input $I$
has the
same output as $\LA$ with input $I'$.
For canonization algorithms as in
\cite{babai1983canonical}
this consistency was achieved by computing
one particular, so called \emph{smallest}, minimal block system~$\CB^*$.
However, such an approach needs some insight into the proof
of \cref{lem:luks}
and one has to argue that the proof can be extended
rather than using the theorem as a black box only.

Using canonical generating sets for objects
over natural numbers, we can reformulate the
theorem. The difference is that in the reformulation the subgroup can be chosen
in a unique way.
This method is
quite general as it indeed uses \cref{lem:luks}
as a black box only.

\begin{theo}[\cite{babai1983canonical} and
\cite{MR679977},\cite{DBLP:conf/focs/BabaiKL83}]
\label{lem:luksN}
Let $\Delta^\ord\leq\sym(\{1,\ldots,|V|\})$
be a group
that is transitive on a set $A^\ord\subseteq \{1,\ldots,|V|\}$.
There is a subgroup $\Psi^\ord\leq\Delta^\ord$
and $b\in\NN$ such that:
\begin{enumerate}
  \item[(1)] The size of the $\Psi^\ord$-orbits on $A^\ord$ are bounded by $b$.  
  \item[(2)] The index $[\Delta^\ord:\Psi^\ord]$ is bounded by
  $(|A^\ord|/b)^{\CO(\cw(\Delta^\ord))}$.
\end{enumerate}
Furthermore, there is an
algorithm $\LB$ that given $A^\ord\subseteq\{1,\ldots,|V|\}$ and a generating set
for $\Delta^\ord$, computes a generating set for the subgroup
$\Psi^\ord\leq\Delta^\ord$.
The group $\Psi^\ord$ only depends on $A^\ord$ and $\Delta^\ord$
(and not  on the given generating set for~$\Delta^\ord$).
The algorithm runs in time $|V|^{\CO(\cw(\Delta^\ord))}$.
\end{theo}

\begin{proof}
Let~$(\delta^\ord_1,\ldots,\delta^\ord_t)$ be a canonical generating set
of~$\Delta^\ord$ ordered according to~``$\prec$''.
We execute the algorithm $\LA$ from
\cref{lem:luks} with input $(A^\ord,(\delta^\ord_1,\ldots,\delta^\ord_t))$
and return its output.
\end{proof}

Using the result of the previous theorem, we can
achieve a running time for canonization expressible in terms of
the composition width of $\Delta$.

\begin{cor}\label{cor:gammaD}
A function $\canObj$
solving \cref{prob:CL:Obj} can be computed in time
$k^{\CO(\cw(\Delta))}n^{\CO(1)}$
where $n$ is the input size and
$k$ is the size of the largest $\Delta$-orbit of $V$.
\end{cor}
\begin{proof}
In the transitive case of the algorithms for $\canMatch$, $\canSetHyper$ and
$\canSet$, we are in the situation of a group $\Delta^\ord$
(or $\Theta^\ord$ respectively)
that is transitive on a set $A^\ord$.
To achieve the running time, we replace the
line
``Define $\Psi^\ord:=\stab_{\Delta^\ord}(A^\ord_1,A^\ord_2)$''
with
``Define $\Psi^\ord$ as the output of algorithm $\LB$
from \cref{lem:luksN}
with input $A^\ord$ and $\Delta^\ord$''.

\begin{runtime}
The critical point for the running time analysis
is that a recursion of the type
$R(|A^*|)\leq (|A^*|/b)^{\CO(\cw(\Delta))} R(b)$
leads to some function $R(|A^*|)\in|A^*|^{\CO(\cw(\Delta))}$.
\end{runtime}
\end{proof}

We remark that there is a particular technique~\cite{DBLP:conf/focs/BabaiKL83}
to improve the running time for isomorphism algorithms from
$k^{\CO(\cw(\Delta))}n^{\CO(1)}$ to $k^{\CO(\cw(\Delta)/\log
\cw(\Delta))}n^{\CO(1)}$.
As this technique exploits the precise structure of large primitive groups and
needs some detailed background, we  will not describe it here.

%% file: articles/openQuestions.tex
\section{Outlook and Open Questions}

We presented a general framework to devise canonization algorithms for
combinatorial objects. It allows for the use of various other algorithms developed in the context of isomorphism problems, without having to worry about isomorphism invariance.
We believe that it not only should be possible to use the framework to design
further canonization algorithms,  but that the framework will
simplify the task. 
Especially the use of canonical generating sets
simplifies the task to adapt theorems used for isomorphism to
adequate canonization variants.
Naturally, the bounded degree isomorphism algorithm of \cite{bounded-degree}
should then be amenable to canonization. However, this remains as
future work.

We ask whether canonization of combinatorial objects can be performed in
time $n^{\polylog(|V|)}$ where~$n$ is the size of the object.
Such a running time is of deep interest
for a canonization algorithm for graphs
of tree width at most $k$
running in time $|V|^{\polylog(k)}$ which
in turn would generalize Babai's
quasipolynomial time bound.
However, even for the isomorphism-version of this problem we have no such algorithm yet.

A different question to this end
is whether the
time bound that we presented in
this paper (of $2^{\CO(|V|)}n^{\CO(1)}$ for objects of size $n$)
can be improved to
moderately exponential
$2^{\CO(|V|^{1-\epsilon})}n^{\CO(1)}$ for some $\epsilon>0$.

It remains an open problem whether isomorphism of permutation groups over a
permutation domain $V$ (that however are implicitly given via a generating set)
can be decided in time $2^{\CO(|V|)}$ regardless of the order of the groups
\cite{DBLP:conf/soda/BabaiCGQ11,codenotti2011testing}. Our framework might help
to design such an isomorphism algorithm.
Indeed, in
Section~\ref{sec:obj}, we showed how codes and permutation groups can be encoded
as combinatorial objects. However, we did so by explicitly listing all
elements.
To encode implicitly given permutation groups, our framework can readily be
extended to allowing implicitly represented
cosets of $V$ to the set $V$ itself
as a new third form of atom,
but at this point we do not see how to handle the arising objects efficiently.